%% file: NPJ-QI-ETP_Arxiv.tex
\documentclass[10pt, aps, pra, twocolumn,nofootinbib,floatfix,showpacs,superscriptaddress,shortbibliography]{revtex4-2}

\usepackage[utf8]{inputenc}

\usepackage{graphicx,epsfig,color, colordvi}
\usepackage{pifont,bbding,amssymb,soul,bbm}
\usepackage{amsmath,amsthm, enumerate, mathtools,slashed, enumitem}
\usepackage{hyperref}

\usepackage{dcolumn}
\usepackage{bm}
\usepackage[dvipsnames]{xcolor}
\usepackage{relsize, accents, wasysym}
\usepackage[capitalize]{cleveref}
\usepackage{times}
\usepackage{comment}
\usepackage{setspace}

\usepackage{tikz-cd}
\tikzcdset{every label/.append style = {font = \small}}

\usepackage{tikz}
\usetikzlibrary{shapes.geometric}
\usetikzlibrary{shapes.arrows}
\usetikzlibrary {arrows.meta}

\hypersetup{
	colorlinks=true,
	linkcolor=CitingColor,
	citecolor=CitingColor,
	urlcolor=CitingColor
}

\DeclareMathAlphabet\mathbfcal{OMS}{cmsy}{b}{n}

\newcommand{\ket}[1]{\ensuremath{|#1\rangle}}
\newcommand{\bra}[1]{\ensuremath{\langle #1|}}
\newcommand{\braket}[2]{\langle #1|#2\rangle}
\newcommand{\proj}[1]{\ket{#1}\!\bra{#1}}
\newcommand{\ketbra}[2]{\ket{#1}\! \bra{#2}}
\newcommand{\vect}[1]{\mathrm{vec}}
\newcommand{\tr}{{\rm tr}}
\newcommand{\innerprod}[2]{\left\langle #1 , #2 \right\rangle}

\newcommand{\be}{\begin{equation}}
\newcommand{\ee}{\end{equation}}
\newcommand{\ba}{\begin{eqnarray}}
\newcommand{\ea}{\end{eqnarray}}

\newcommand{\calS}{\mathcal{S}}

\newcommand{\W}{\mathcal{W}}

\newcommand{\ra}{{\rm A}}
\newcommand{\rb}{{\rm B}}
\newcommand{\rc}{{\rm C}}
\newcommand{\rd}{{\rm D}}
\newcommand{\rabc}{{\rm ABC}}
\newcommand{\rab}{{\rm AB}}
\newcommand{\rac}{{\rm AC}}
\newcommand{\rbc}{{\rm BC}}

\newcommand{\rms}{{\rm S}}
\newcommand{\rmt}{{\rm T}}
\newcommand{\Tp}{^{\mbox{\tiny \rm T}}}

\newcommand{\norm}[1]{\left\|#1\right\|}

\newcommand{\id}{\mathbb{I}}

\newcommand{\rhos}{\rho_\rms}
\newcommand{\rhot}{\rho_\rmt}

\newtheorem{theorem}{Theorem}

\newtheorem{lemma}[theorem]{Lemma}

\newtheorem{question}{Question}

\definecolor{nred}{rgb}{0.9,0.1,0.1}
\definecolor{nblack}{rgb}{0,0,0}
\definecolor{nblue}{rgb}{0.2,0.2,0.8}
\definecolor{ngreen}{rgb}{0.2,0.5,0.2}
\definecolor{ublue}{rgb}{0,0,0.5}

\definecolor{pur}{rgb}{0.3,0.1,0.6}
\definecolor{nngrn}{rgb}{0,0.5,0.5}
\definecolor{CitingColor}{rgb}{0,0.3,1}

\newcommand{\blu}{\color{nblue}}

\begin{document}
\title{Entanglement transitivity problems}

\author{Gelo Noel M. Tabia}
\email{gelonoel-tabia@gs.ncku.edu.tw}
\affiliation{Department of Physics and Center for Quantum Frontiers of Research \& Technology (QFort), National Cheng Kung University, Tainan 701, Taiwan}
\affiliation{Physics Division, National Center for Theoretical Sciences, Taipei 10617, Taiwan}
\affiliation{Center for Quantum Technology, National Tsing Hua University, Hsinchu 300, Taiwan}

\author{Kai-Siang Chen}
\affiliation{Department of Physics and Center for Quantum Frontiers of Research \& Technology (QFort), National Cheng Kung University, Tainan 701, Taiwan}

\author{Chung-Yun Hsieh}
\affiliation{ICFO - Institut de Ci\`encies Fot\`oniques, The Barcelona Institute of Science and Technology, 08860 Castelldefels, Spain}

\author{Yu-Chun Yin}
\affiliation{Department of Physics and Center for Quantum Frontiers of Research \& Technology (QFort), National Cheng Kung University, Tainan 701, Taiwan}

\author{Yeong-Cherng Liang}
\email{ycliang@mail.ncku.edu.tw}
\affiliation{Department of Physics and Center for Quantum Frontiers of Research \& Technology (QFort), National Cheng Kung University, Tainan 701, Taiwan}
\affiliation{Physics Division, National Center for Theoretical Sciences, Taipei 10617, Taiwan}

\date{\today}
\begin{abstract}
One of the goals of science is to understand the relation between a whole and its parts, as exemplified by  the problem of certifying the entanglement of a system from the knowledge of its reduced states.  Here, we focus on a different but related question: can a collection of marginal information reveal new marginal information?  We answer this affirmatively and show that (non-) entangled marginal states may exhibit {\em (meta)transitivity of entanglement}, i.e., implying that a different target marginal must  be entangled.  By showing that the global $n$-qubit state compatible with certain two-qubit marginals in a tree form is unique, we prove that transitivity exists for a system involving an arbitrarily large number of qubits. We also completely characterize---in the sense of providing both the necessary and sufficient conditions---when (meta)transitivity can occur in a tripartite scenario when the two-qudit marginals given are either the Werner states or the isotropic states. Our numerical results suggest that in the tripartite scenario, entanglement transitivity is {\em generic} among the marginals derived from pure states.
\end{abstract}
 
\maketitle

\section*{Introduction}
Entanglement~\cite{horodeckicubed2009} is a  characteristic of quantum theory that profoundly distinguishes it from classical physics. The modern perspective considers entanglement as a resource for information processing tasks, such as quantum computation~\cite{nielsen00,Jozsa:2003wj,Vidal.PRL.2003,Preskill2018,McArdle2020}, quantum simulation~\cite{Georgescu2014}, and quantum metrology~\cite{Pezze2018}.  
With the huge effort devoted to scaling up quantum technologies~\cite{Ladd2010}, considerable attention has been given to the study of quantum many-body systems~\cite{Gabriele2018,cirac2012entanglement}, specifically the ability to prepare and manipulate large-scale entanglement in various experimental systems. 

As the number of parameters to be estimated is huge, entanglement detection via the so-called state tomography is often impractical.
Indeed, significant efforts have been made for detecting entanglement in many-body systems~\cite{Gabriele2018,cirac2012entanglement} using limited marginal information. 
For example, some tackle the problem using  properties of the reduced states~\cite{Toth2005-PRA,NavascuesOwariPlenio2009,JungnitschMoroderGuhne2011,Sawicki2012PRA,SperlingVogel2013,WalterDoranGrossChristandl2013,Chen.PRA.2014,MiklinMoroderGuhne2016,BohnetWaldraffBraunGiraud2017,HarrowNatarajanWu2017,SperlingVogel2018,Paraschiv2018},
while others exploit directly the data from local measurements~\cite{TothKnappGunheBriegel2007-PRL,TothKnappGuhneBriegel2009, Guhne:2009wn,GittsovichHyllusGuhne2010,deVicenteHuber2011,Bancal.PRL.2011,LiWangFeiJost2014,Liang.PRL.2015,BaccariCavalcantiWittekAcin2017,He:PRX:2018,FrerotRoscilde2021, FrerotBaccariAcin2021}.
Despite their differences, they can all be seen as some kind of {\em entanglement marginal problem} (EMP)~\cite{NavascuesBaccariAcin2021}, where the entanglement of the global system is to be deduced from some (partial knowledge of the) reduced states.

The entanglement of the global system, nonetheless, is not always the desired quality of interest. For instance, in scaling up a quantum computer, one may wish to verify that a specific subset of qubits indeed get entangled, but this generally does not follow from the entanglement of the global state (recall, e.g., the Greenberger-Horne-Zeilinger states~\cite{GHZ}).
Thus, one requires a more general version of the problem: Given certain reduced states,  can we certify the entanglement in some {\em other} target (marginal) state? We call this  the {\em entanglement transitivity problem} (ETP). Since the global system is a legitimate target system,  ETPs include the EMP as a special case.

As a concrete example beyond EMPs, one may wonder whether a set of {\em entangled} marginals are {\em sufficient} to guarantee the entanglement of some {\em other} target subsystems. If so, inspired by the work~\cite{corettihanggiwolf2011} on {\em nonlocality transitivity} of {\em post}-quantum correlations~\cite{Popescu1994}, we say that such marginals exhibit \emph{entanglement transitivity}. Indeed, one of the motivations for considering entanglement transitivity is that it is a prerequisite for the nonlocality transitivity of {\em quantum correlations}, a problem that has, to our knowledge, remained open.

More generally, one may also wonder whether {\em separable} marginals {\em alone}, or with some entangled marginals could imply the entanglement of other marginal(s). To distinguish this from the above phenomenon, we say that such marginals exhibit {\em metatransitivity}. Note that any instance of metatransitivity with only separable marginals represents a positive answer to the EMP. Here, we show that examples of both types of transitivity can indeed be found. Moreover, we completely characterize when two Werner-state~\cite{werner1989} marginals and two isotropic-state~\cite{Horodecki.isotropic} marginals may exhibit (meta)transitivity.

\section*{Results}

\subsection*{Formulation of the entanglement transitivity problems} 

Let us first stress that in an ETP, the set of given reduced states {\em must be} compatible, i.e., giving a positive answer to the quantum marginal problem~\cite{klyachko2006,TycVlach2015}. With some thought, one realizes that the simplest nontrivial ETP involves a three-qubit system where two of the two-qubit marginals are provided. Then, the problem of deciding if the remaining two-qubit marginal {\em can be} separable is an ETP  different from EMPs. 

More generally,  for any $n$-partite system $\rms$, an instance of the ETP is defined by specifying a set $\calS= \{ \rms_{i}: i = 1,2,\ldots,k\}$ of $k$ marginal systems $\rms_{i}$ (each in its respective state $\sigma_{\rms_{i}}$) and a target system $\rmt\not\in\calS$. Here, $\calS$ is a {\em strict} subset of all the $2^n$ possible combinations of at most $n$ subsystems, i.e., $k<2^n$. 
Then, $\bm{\sigma}:=\{\sigma_{\rms_{i}}\}$ exhibits entanglement (meta)transitivity in $\rmt$ if {\em for all} joint states $\rhos$ compatible with $\bm{\sigma}$, the reduced state $\rhot$ is {\em always} entangled while (not) all given $\sigma_{\rms_{i}}$ are entangled. Formally, the compatible requirement reads as: $\tr_{\rms\backslash \rms_{i}}(\rhos) = \sigma_{\rms_{i}}$ for all $\rms_{i}\in\calS$ where $\rms\backslash \rms_{i}$ denotes the complement of $\rms_{i}$ in the global system $\rms$.

Notice that for the problem to be nontrivial, there must be (1) some overlap among the subsystems specified by $\rms_i$'s, as well as with $\rmt$, and (2) the global system $\rms$ cannot be a member of $\calS$. However, the target system $\rmt$ may be chosen to be $\rms$ and if {\em all} $\sigma_{\rms_{i}}$ are separable, we recover the  EMP~\cite{NavascuesBaccariAcin2021} (see also~\cite{MiklinMoroderGuhne2016,Paraschiv2018} for some strengthened version of the EMP). Hereafter, we focus on ETPs beyond EMPs, albeit some of the discussions below may also find applications in EMPs.

\subsection*{Certification of (meta)transitivity by a linear witness}

Let $\mathcal{W}(\rho)$ be an entanglement witness~\cite{Guhne:2009wn}, i.e.,  $\mathcal{W}(\rho) \ge 0$ for all separable states in $\rmt$, and $\mathcal{W}(\rho)< 0$ for some entangled states. We can  certify the (meta)transitivity of $\calS$ in $\rmt$ if a negative optimal value is obtained for the following optimization problem:
\begin{align}
\label{eq.transitivityOpt}
\max_{\rhos} \W(\rhot), 
\text{ s.t. }  \tr_{\rms\backslash \rms_{i}}(\rhos) = \sigma_{\rms_{i}} \forall\,\, \rms_{i}\in\calS, \,\, \rhos \succeq 0,
\end{align}
where $\tr(\rhos)=1$ is implied by the compatibility requirement and ``$\succeq$" denotes matrix positivity.
Then, $\mathcal{W}$ detects the entanglement in $\rmt$ from the given marginals in $\calS$.

Consider now a linear entanglement witness, i.e., $\mathcal{W}(\rhot) = \tr\left[\rhos(W_{\rmt}\otimes\mathbb{I}_{\rms\backslash\rmt})\right]$ for some Hermitian operator $W_{\rmt}$, where $\rhot=\tr_{\rms\backslash\rmt}(\rhos)$ is the reduced state of $\rho$ in $\rmt$.
In this case, \cref{eq.transitivityOpt} is a semidefinite program~\cite{boyd2004convex}. Interestingly,  its dual problem~\cite{boyd2004convex} can be seen as the problem of minimizing the total interaction energies among the subsystems $\rms_i$ while ensuring that the global Hamiltonian is non-negative, see Supplementary Note~1.

Hereafter, we focus, for simplicity, on $\rmt$ being a two-body system. Then, a convenient witness is that due to the positive-partial-transpose (PPT)
 criterion~\cite{peres1996,horodecki1996}, with $W_\rmt = \eta_\rmt^\Gamma$, where $\eta_\rmt\succeq 0$ and $\Gamma$ denotes the partial transposition operation. Further minimizing the optimum value of \cref{eq.transitivityOpt} over all $\eta_\rmt$ such that $\tr(\eta_\rmt)=1$ gives an optimum $\lambda^*$ that is provably the smallest eigenvalue of all compatible $\rhot^\Gamma$ (see Supplementary Note~1). Hence, $\lambda^*<0$ is a sufficient condition for witnessing the entanglement (meta)transitivity of the given $\bm{\sigma}$ in $\rmt$. 

Three remarks are now in order. Firstly, the ETP defined above is straightforwardly generalized to include {\em multiple} target systems $\{ \rmt_{j}: j = 1,\ldots,t \}$ with $\rmt_j\not\in\calS$ for all $j$. A certification of the joint (meta)transitivity is then achieved by certifying each $\rmt_j$ separately. Secondly, other entanglement witnesses~\cite{Guhne:2009wn} may  be considered. For instance, to certify the entanglement of a two-body $\rhot$ that is PPT~\cite{Horodecki:PRL:1998}, a witness based on the computable cross-norm/ realignment (CCNR) criterion~\cite{Rudolph2005,ChenWu2003,Horodecki2006, Shang2018}, 
may be employed. Finally, for a multipartite target system, a witness tailored for detecting the genuine multipartite entanglement in $\rhot$ (see, e.g., Ref.~\cite{JungnitschMoroderGuhne2011,SperlingVogel2013}) is surely of interest.

\begin{figure}[t!]
    \centering
\begin{tikzpicture}
\begin{scope}[every node/.style={circle,fill,inner sep=0pt, minimum size = 1pt, scale = 0.7}]
    \node (A) at (0,0.3) {A};
    \node (B) at (0.75,0.3) {B};
    \node (C) at (1.5,0.3) {C};
    \node (D) at (2.25,0.3) {D};
\end{scope}

\begin{scope}
    \node at (0, 0.7) {$\ra$};
    \node at (0.75, 0.7) {$\rb$};
    \node at (1.5, 0.7) {$\rc$};
    \node at (2.25, 0.7) {$\rd$};
\end{scope}

\begin{scope}[every node/.style={fill=white,circle},
              every edge/.style={draw=black,very thick}]
    \path [-] (A) edge (B);
    \path [-] (B) edge (C);
    \path [-] (C) edge (D);
\end{scope}
\begin{scope}[every node/.style={circle,fill,inner sep=0pt, minimum size = 1pt, scale = 0.7}]
    \node (A) at (3.5,0) {A};
    \node (B) at (4.25,0) {B};
    \node (C) at (5.0,0) {C};
    \node (D) at (4.25,0.75) {D};
\end{scope}
\begin{scope}
    \node at (3.1, 0) {$\rb$};
    \node at (4.25, -0.4) {$\rc$};
    \node at (5.4, 0) {$\rd$};
    \node at (4.25, 1.15) {$\ra$};
\end{scope}

\begin{scope}[every node/.style={fill=white,circle},
              every edge/.style={draw=black,very thick}]
    \path [-] (A) edge (D);
    \path [-] (B) edge (D);
    \path [-] (C) edge (D);
\end{scope}

\begin{scope}[every node/.style={circle,fill,inner sep=0pt, minimum size = 1pt, scale = 0.7}]
    \node (A) at (6.25,0) {A};
    \node (B) at (6.25,0.75) {B};
    \node (C) at (7, 0.75) {C};
    \node (D) at (7,0) {D};
\end{scope}
\begin{scope}[every node/.style={fill=white,circle},
              every edge/.style={draw=black,very thick}]
    \path [-] (A) edge (B);
    \path [-] (B) edge (C);
    \path [-] (C) edge (A);
\end{scope}

\begin{scope}
    \node at (1.125,-0.9)  {(a)};
    \node at (4.25,-0.9)  {(b)};
     \node at (6.625,-0.9)  {(c)};
\end{scope}
\end{tikzpicture}

    \caption{Tree graph. A tree graph is any undirected acyclic graph such that a unique path connects any two vertices. Graph (a) and (b) are the only two nonisomorphic trees with $(n-1)$ edges for $n = 4$. Graph (c) is not a tree because it is disconnected and has a cycle.}
    \label{fig.TreeGraph}
\end{figure}
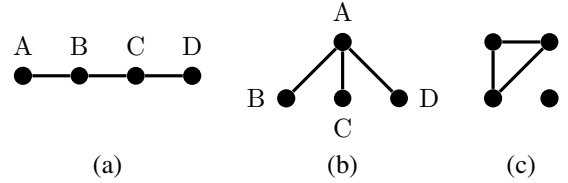

\subsection*{A family of transitivity examples with \texorpdfstring{$n$}{n} qubits}

As a first illustration, let $\ket{\Psi^+}=\tfrac{1}{\sqrt{2}}(\ket{10} +\ket{01})$ and consider: 
\begin{equation}
    \rho_{n}(\gamma) = \left(\tfrac{n-2\gamma}{n} \right)
    \proj{00} + \tfrac{2\gamma}{n}\proj{\Psi^+}, \quad n\ge 3,
\end{equation}
which is a two-qubit reduced state of
$\Omega_{n}(\gamma) = \gamma\proj{W_n} + (1-\gamma)\proj{0^n}$,
i.e., a mixture of $\ket{0^n}$ and an $n$-qubit $W$ state
$\ket{W_n} = \tfrac{1}{\sqrt{n}} \sum_{j=1}^{n}\ket{1_{j}}$,
where $1_j$ denotes an $n$-bit string with a 1 in position $j$ and 0 elsewhere.
Now, imagine drawing these $n$ qubits as vertices of a tree graph~\cite{BenderWilliamson2010} with $(n-1)$ edges, see Fig.~\ref{fig.TreeGraph}, such that every edge corresponds to a pair of qubits in the state $\rho_{n}(\gamma)$, that is, 
\begin{equation}\label{Eq:Wstate_marginal}
    \tr_{\rms\backslash \rms_i} (\rho) = \sigma_{\rms_i}=\rho_n(\gamma) \,\,\forall\,\, \rms_i \in \calS,
\end{equation}
where $\calS$ represents the set of edges. Then we prove the following result:
\begin{theorem}\label{Thm.uniqueness}
For any tree graph with $n$ vertices that satisfies ~\cref{Eq:Wstate_marginal}, $\Omega_{n}(\gamma) = \gamma\proj{W_n} + (1-\gamma)\proj{0^n}$ is the unique global state and all the two-qubit reduced states are $\rho_n(\gamma)$.
\end{theorem}
The details of its proof can be found in Supplementary Notes 2.
Thus, these $\rho_{n}(\gamma)$ exhibit transitivity for any of the $\frac{(n-1)(n-2)}{2}$ pairs of qubits that are {\em not} linked by an edge. Indeed, the symmetry of $\Omega_{n}(\gamma)$ implies that all its two-qubit marginals are $\rho_{n}(\gamma)$, and the smallest eigenvalue of $\rho_{n}(\gamma)^\Gamma$ is 
$\lambda^*=\tfrac{(n-2\gamma)-\sqrt{(n-2\gamma)^2+4\gamma^2}}{2n}<0$ for $\gamma\in(0,1]$.

We should clarify that the transitivity exhibited by $\rho_n(\gamma)$ requires a tree graph only in that it represents the minimal amount of marginal information for the global state to be uniquely determined. Any other $n$-vertex graph with equivalent marginal information or more leads to the same conclusion.

These examples involve only entangled marginals. Next, we present examples where some of the given marginals are separable. In particular, we provide a {\em complete} solution of the ETPs with the input marginals being  a Werner state~\cite{werner1989} or an isotropic state~\cite{Horodecki.isotropic}.

\begin{figure}
    \centering
    \includegraphics[width=0.65\linewidth]{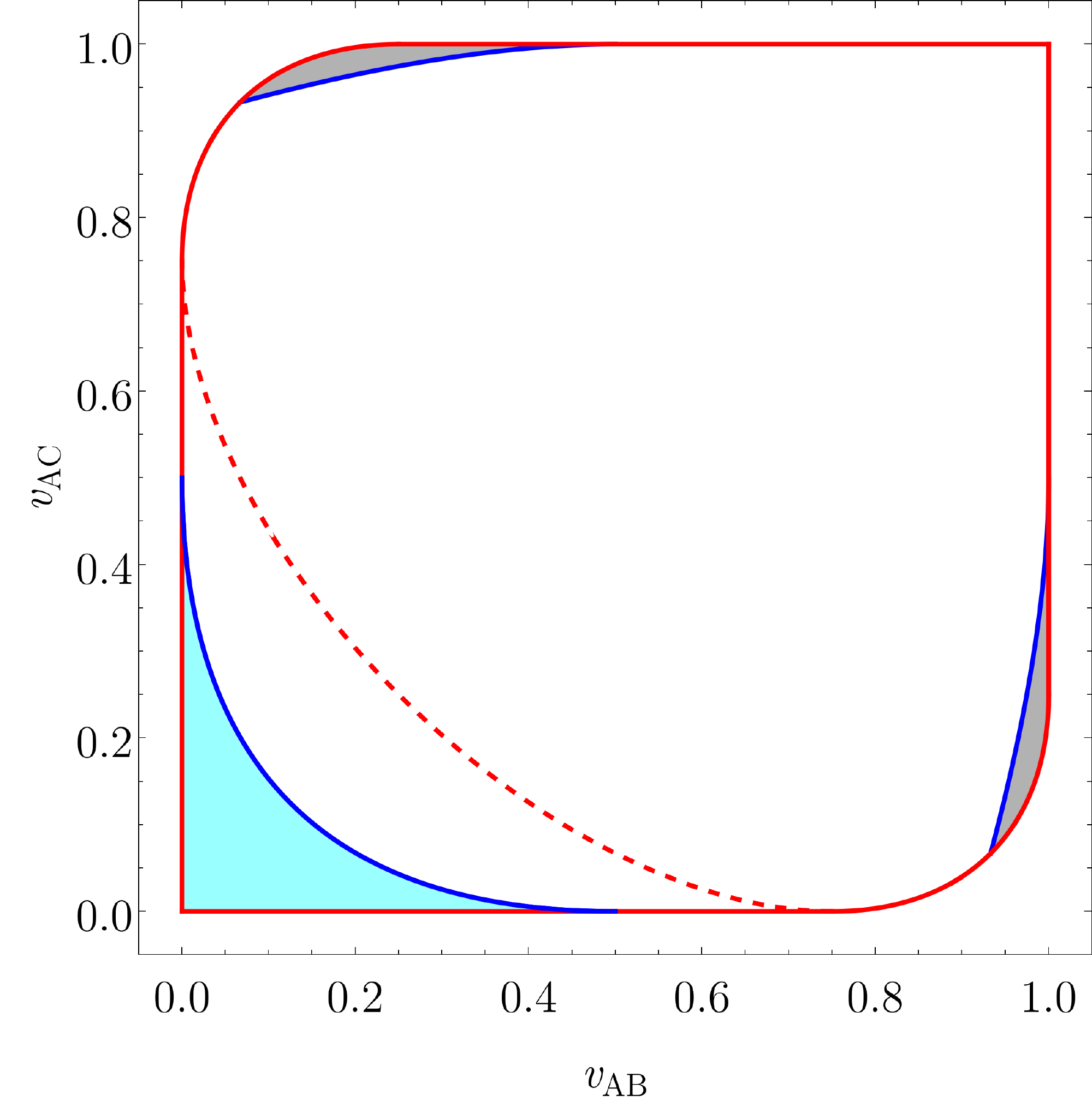}
    \caption{Parameter space for a pair of Werner state marginals with, respectively, weight $v_{\rab}$ and $v_{\rac}$ on the symmetric subspace. 
    For $d\ge 3$  the compatible region for the pair is enclosed by the solid red line, but for $d=2$ it is restricted to the portion above the dotted line.
    The blue curves (being parts of two parabolas) describe boundaries where the largest compatible $v_{\rbc}$ is $\tfrac{1}{2}$. Regions exhibiting (meta)transitivity are shaded in (gray) cyan. 
    }
    \label{fig:MetatransitivityQuditWerner}
\end{figure}

\subsection*{Metatransitivity from Werner state marginals}
A Werner state~\cite{werner1989} $W_d(v)$ is a two-qudit density operator invariant under arbitrary $U\otimes U$ unitary transformations, where $U$ belongs to the set of $d$-dimensional unitaries $\mathcal{U}_d$ for finite $d$. Let $P^{d}_\mathrm{s} (P^{d}_\mathrm{as}) $ be the projection onto the symmetric (antisymmetric) subspace of $\mathbb{C}^d\otimes\mathbb{C}^d$. Then we can write qudit Werner states as the one-parameter family~\cite{werner1989}
\begin{align}\label{Eq:Werner}
    W_d(v) &= v \tfrac{2}{d(d+1)}P^{d}_\mathrm{s} + (1-v)\tfrac{2}{d(d-1)} P^{a}_\mathrm{as}, &
     v\in [0,1].
\end{align}
Consider a pair of Werner states $\bm{\sigma}=\{W_d(v_{\rab}), W_d(v_{\rac})\}$ that are the marginals of some joint state $\rho_{\rabc}$. Then the Werner-twirled state~\cite{EggelingWerner2001}
$\widetilde{\rho}_{\rabc} = \int d\mu_U (U\otimes U \otimes U) \rho_{\rabc} (U\otimes U \otimes U)^\dag, $
where $\mu_U$ is a uniform Haar measure over $\mathcal{U}_d$, is trivially verified to be a valid joint state for these marginals. Moreover, $\widetilde{\rho}_{\rabc}$ has a Werner state $W_d(v_{\rbc})$ as its BC marginal. 

Importantly, the  aforementioned twirling bringing $\rho_{\rabc}$ to $\widetilde{\rho}_{\rabc}$ is achievable by local operations and classical communications (LOCC). Since LOCC cannot create entanglement from none, if the BC marginal $\widetilde{\rho}_{\rbc}$ of $\widetilde{\rho}_{\rabc}$ is entangled, so must the BC marginal $\rho_{\rbc}$ of $\rho_{\rabc}$. Conversely, since $\widetilde{\rho}_{\rabc}$ is a legitimate joint state of the given  marginals $\bm{\sigma}$, if $\widetilde{\rho}_{\rbc}$ is separable, by definition, the given marginals $\bm{\sigma}$ cannot exhibit transitivity. Without loss of generality, we may thus restrict our attention to  a Werner-twirled joint state $\widetilde{\rho}_{\rabc}$. Then, since a Werner state  $W_d(v)$ is entangled {\em if and only if} ({\em iff})~\cite{werner1989} $v\in[0,\tfrac{1}{2})$, combinations of Werner state marginals $W_d(v_{\rab})$ and $W_d(v_{\rac})$ leading to  $\widetilde{\rho}_{\rbc}=W_d(v_{\rbc})$ with $v_{\rbc}<\tfrac{1}{2}$ must exhibit entanglement (meta)transitivity.

Next, let us recall from Ref.~\cite{johnsonviola2013} the following characterization:  three Werner states with parameters $\vec{v} = (v_{\rab}, v_{\rac}, v_{\rbc})$ are compatible {\em iff} the vector $\vec{v}$ lies within the bicone given by
$f(\vec{v}) \ge g(\vec{v})$
and $3 - f(\vec{v})  \ge g(\vec{v})$,
where 
$f(\vec{v}) = v_{\rab} + v_{\rac} + v_{\rbc}$ and
$g(\vec{v}) =\sqrt{3(v_{\rac} - v_{\rab})^2 + (2v_{\rbc} - v_{\rab} - v_{\rac})^2}$.
To find the (meta)transitivity region for $(v_{\rab},v_{\rac})$, it suffices to determine the boundary where the largest compatible $v_{\rbc} = \tfrac{1}{2}$. These boundaries are found (see Supplementary Note~3) to be the two parabolas
$(v_{\rab} - v_{\rac}-\tfrac{1}{2})^2 = 2(1-v_{\rab})$ and
$(v_{\rab} + v_{\rac}-\tfrac{1}{2})^2 = 4v_{\rab}v_{\rac}$, mirrored along the line $v_{\rab} + v_{\rac} = 1$, as shown in ~\cref{fig:MetatransitivityQuditWerner}. It also shows the compatible regions of $(v_{\rab},v_{\rac})$ obtained directly from Ref.~\cite{johnsonviola2013}, and the desired (shaded) regions exhibiting the (meta)transitivity of these marginals. 
In particular, the lower-left region corresponds to (a) while the top-left and bottom-right regions correspond to (b) in Fig.~\ref{fig:MetaEx}.
Remarkably, these results hold for arbitrary Hilbert space dimension $d\ge 2$ (but for $d=2$, the lower-left shaded region does not correspond to compatible Werner marginals).

\begin{figure}[t!]
    \centering
    \includegraphics[width=0.65\linewidth]{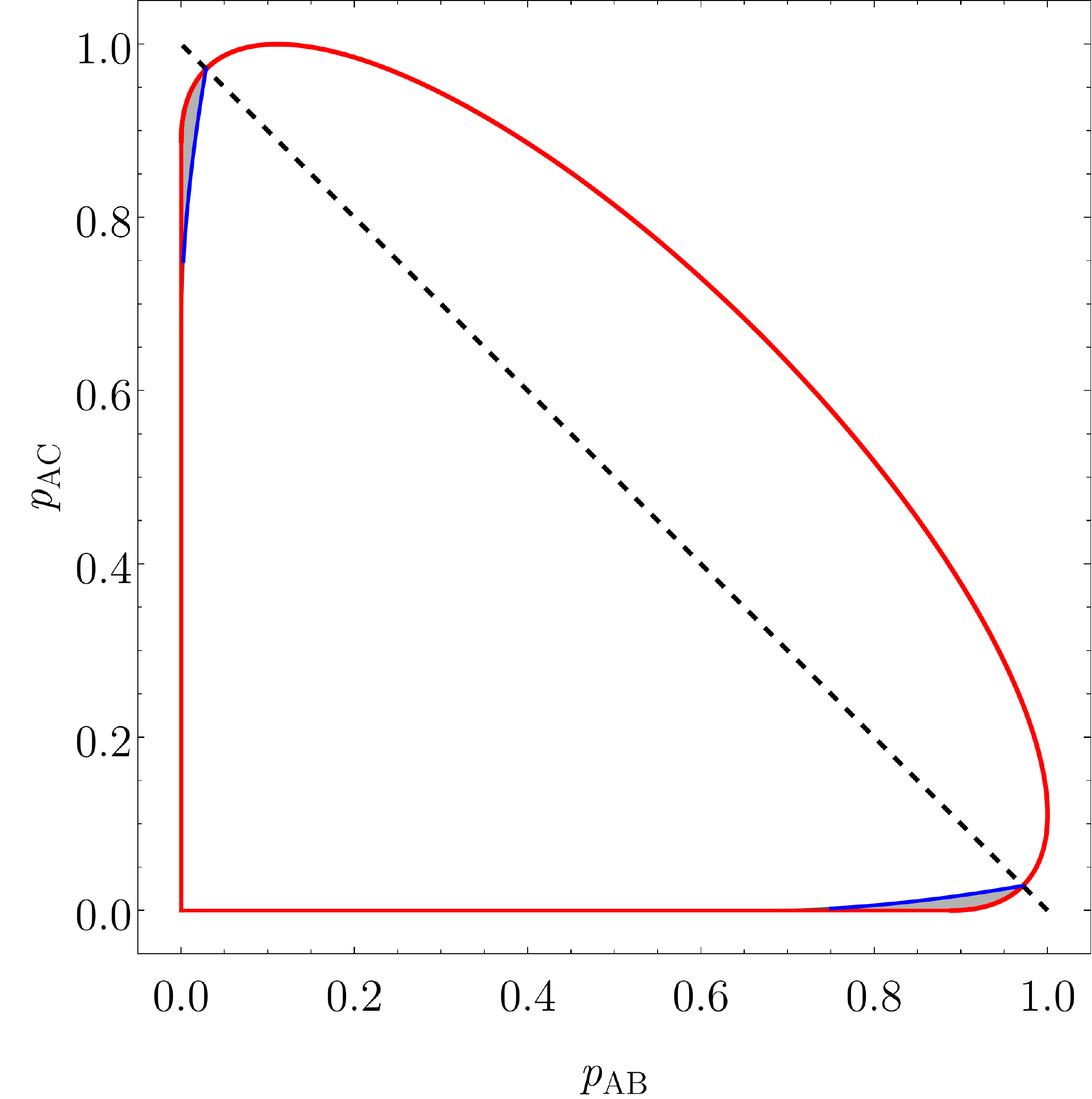}
    \caption{Parameter space for a pair of isotropic state marginals with, respectively, fully entangled fraction $p_{\rab}$ and $p_{\rac}$. 
    The compatible region for the pair is enclosed by the solid red line. The blue curves (shown for the case of $d=3$) marks the boundary where the largest compatible $v_{\rbc}$ is $\tfrac{1}{2}$. Regions exhibiting metatransitivity are shaded in gray, which shrink with increasing $d$, as the upper red curve flattens towards the dashed black line and the blue curves approach the two axes.} 
    \label{fig:MetatransitivityQuditIsotropic}
\end{figure}

\subsection*{Metatransitivity from isotropic state marginals}
An isotropic state~\cite{Horodecki.isotropic} is a bipartite density operator in $\mathbb{C}^d\otimes\mathbb{C}^d$ that is invariant under $U\otimes \overline{U}$ (or $\overline{U}\otimes U$) transformations for any unitary $U \in \mathcal{U}_d$; here, $\overline{U}$ is the complex conjugation of $U$. We can write qudit isotropic states as a one-parameter family~\cite{Horodecki.isotropic}
\begin{equation}
    \mathcal{I}_d(p) = p \proj{\Phi_d} +
    \tfrac{1-p}{d^2-1} \left(\mathbb{I}_{d^2} - \proj{\Phi_d} \right),
\end{equation}
where 
$\ket{\Phi_d} = \tfrac{1}{\sqrt{d}}\sum_{j=0}^{d-1} \ket{j}\ket{j}$ 
and $p$ gives the fully entangled fraction~\cite{Horodecki1999,Albeverio2002} of $\mathcal{I}_d(p)$.

Consider now a pair of isotropic marginals $\bm{\sigma}=\{\mathcal{I}_d(p_{\rab}), \mathcal{I}_d(p_{\rac})\}$ as the reduced states of some joint state $\tau_{\rabc}$. Then the ``twirled'' state 
$\widetilde{\tau}_{\rabc} = \int d\mu_U (\overline{U}\otimes U \otimes U) \tau_{\rabc} (\overline{U}\otimes U \otimes U)^\dag $, which has a Werner state marginal $W_d(v_{\rbc})$ in BC,
is easily verified to be a valid joint state for the given marginals. As in the case of given Werner states marginals,  it suffices to consider $\widetilde{\tau}_{\rabc}$ in determining the region of $(p_{\rab},p_{\rac})$ that demonstrates metatransitivity. 

To this end, note that two isotropic states and one Werner state with parameters $\vec{p} = (p_{\rab},p_{\rac}, v_{\rbc})$ are compatible {\em iff}~\cite{johnsonviola2013} the vector $\vec{p}$ lies within the convex hull of the origin $\vec{p}_0 = (0,0,0)$ and the cone given by 
$\alpha_{+} \le 1 + \tfrac{1}{d}(\beta+1)$
and 
$d\alpha_{+} - \beta \ge d\sqrt{\left( \alpha_{+} + \beta \right)^2 + \left(\tfrac{d+1}{d-1}\right)\alpha_{-}}$,
where $\alpha_{\pm} = p_{\rab} \pm p_{\rac}$
and $\beta = 2(v_{\rbc} -1)$.
To find the metatransitivity region for $(p_{\rab},p_{\rac})$ we again look for the boundary where the largest compatible $v_{\rbc} = \tfrac{1}{2}$, which we show in Supplementary Note~4 to be $4 p_{\rab}p_{\rac} = (p_{\rab} + p_{\rac} -1 + \tfrac{1}{d})^2$. The resulting regions of interest are illustrated for the $d=3$ case in ~\cref{fig:MetatransitivityQuditIsotropic}, and they correspond to (b) in Fig.~\ref{fig:MetaEx}.

\subsection*{Metatransitivity with only separable marginals}
Curiously, none of the infinitely compatible pairs of marginals given above result in the most exotic type of metatransitivity, even though there are known examples where separable marginals imply a global entangled state (see, e.g.,~\cite{Toth2005-PRA,TothKnappGunheBriegel2007-PRL,TothKnappGuhneBriegel2009,NavascuesBaccariAcin2021}). In the following, we provide examples where
the entanglement of a subsystem is implied by only separable marginals. This already occurs in the simplest case of a three-qubit system. 
Consider the rank-two mixed state $\chi_{\rabc} = \tfrac{1}{4}\proj{\chi_{1}} + \tfrac{3}{4}\proj{\chi_{2}}$ where
\begin{align}
\label{eq.SepABSepBCEntAC}
\nonumber
    \ket{\chi_1} &= 
    \begin{pmatrix}
 \frac{1}{3} & \frac{1}{12} & -\frac{\sqrt{7}}{12} & 0 & \frac{\sqrt{7}}{12} & -\frac{1}{3} & -\frac{3}{4} & \frac{1}{3} 
\end{pmatrix}\Tp, \\
    \ket{\chi_2} &= 
    \begin{pmatrix}
 -\frac{1}{2} & \frac{\sqrt{5}}{24} & \frac{1}{6} & \frac{1}{8} & -\frac{1}{3} & -\frac{3}{4} & \frac{\sqrt{5}}{24} & \frac{1}{8} 
\end{pmatrix}\Tp.
\end{align}
It can be easily checked that the AB and BC marginals of $\chi_{\rabc}$ are PPT, which suffices~\cite{horodecki1996} to guarantee their separability, while \cref{eq.transitivityOpt} with the PPT criterion can be used to confirm that AC is always entangled. Thus, this example corresponds to (c) in Fig.~\ref{fig:MetaEx}.
Likewise, examples exhibiting different kinds of transitivity can also found in higher dimensions (with bound entanglement~\cite{Horodecki:PRL:1998}) or with more subsystems, see Supplementary Note~5 for details. 

Here, we present one such example to illustrate some of the subtleties of ETPs in a scenario involving more than three subsystems. Consider the four-qubit pure state
\begin{align}
    \label{eq.SepABSepBCSepCDEntACADBD}
    \nonumber
    \ket{\xi}_{\rm ABCD} &= 
    \left( \tfrac{1}{45},-\tfrac{1}{3},\tfrac{1}{3},\tfrac{1}{9},\tfrac{2}{9},-\tfrac{1}{4},-\tfrac{2}{5},\tfrac{1}{9}, \right. \\
   & \qquad \left.
   \tfrac{\sqrt{10}}{36},\tfrac{1}{9},-\tfrac{1}{9},-\tfrac{1}{4},-\tfrac{1}{2},\tfrac{1}{9},-\tfrac{1}{9},\tfrac{1}{3} \right)\Tp.
\end{align}
One can readily check that its AB, BC, and CD marginals are PPT and are thus separable. At the same time, one can verify using \cref{eq.transitivityOpt} with the PPT criterion that these three marginals together imply the entanglement of {\em all} the three remaining two-qubit marginals. Thus, this corresponds to (d) in Fig.~\ref{fig:MetaEx}.

At this point, one may think that the entanglement in the AC marginal already follows from the given AB and BC marginals, analogous to the tripartite examples presented above. This is misguided: the CD marginal is essential to force the AC marginal to be entangled. Similarly, the AB marginal is indispensable to guarantee the entanglement of BD. Thus, the current metatransitivity example illustrates a genuine four-party effect that cannot exist in any tripartite scenario. For completeness, an example exhibiting the same four-party effect but where all input two-qubit marginals are entangled is also provided in Supplementary Note~5.

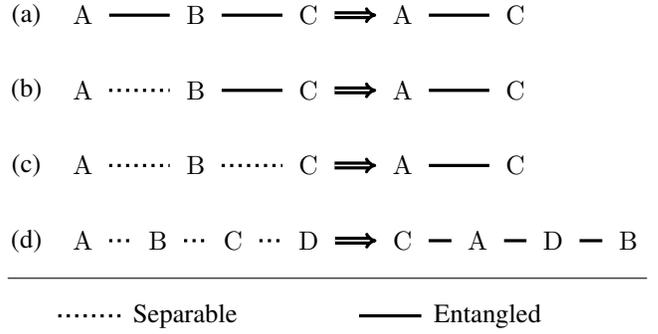
\begin{figure}[t!]
    \centering
\begin{tikzpicture}
\draw (0,5) -- (8.5,5);
\begin{scope}
    \node at (0.25,4.5)  {(a)};
    \node at (0.25,3.5)  {(b)};
    \node at (0.25,2.5)  {(c)};
    \node at (0.25,1.5)  {(d)};
\end{scope}

\begin{scope}[every node/.style={circle}]
    \node (A) at (1,4.5) {$\mathrm{A}$};
    \node (B) at (2.5,4.5) {$\mathrm{B}$};
    \node (C) at (4,4.5) {$\mathrm{C}$};
    \node (D) at (5.25,4.5) {$\mathrm{A}$};
    \node (E) at (6.75,4.5) {$\mathrm{C}$};
\end{scope}
\begin{scope}[every node/.style={fill=white,circle},
              every edge/.style={draw=black,very thick}]
    \path [-] (A) edge (B);
    \path [-] (B) edge (C);
    \path [-] (D) edge (E);
    \path [-](C) edge[double, -{Implies}] (D);
\end{scope}

\begin{scope}[every node/.style={circle}]
    \node (A) at (1,3.5) {$\mathrm{A}$};
    \node (B) at (2.5,3.5) {$\mathrm{B}$};
    \node (C) at (4,3.5) {$\mathrm{C}$};
    \node (D) at (5.25,3.5) {$\mathrm{A}$};
    \node (E) at (6.75,3.5) {$\mathrm{C}$};
\end{scope}
\begin{scope}[every node/.style={fill=white,circle},
              every edge/.style={draw=black,very thick}]
    \path [-] (A) edge[dotted] (B);
    \path [-] (B) edge (C);
    \path [-] (D) edge (E);
    \path [-](C) edge[double, -{Implies}] (D);
\end{scope}

\begin{scope}[every node/.style={circle}]
    \node (A) at (1,2.5) {$\mathrm{A}$};
    \node (B) at (2.5,2.5) {$\mathrm{B}$};
    \node (C) at (4,2.5) {$\mathrm{C}$};
    \node (D) at (5.25,2.5) {$\mathrm{A}$};
    \node (E) at (6.75,2.5) {$\mathrm{C}$};
\end{scope}
\begin{scope}[every node/.style={fill=white,circle},
              every edge/.style={draw=black,very thick}]
    \path [-] (A) edge[dotted] (B);
    \path [-] (B) edge[dotted] (C);
    \path [-] (D) edge (E);
    \path [-](C) edge[double, -{Implies}] (D);
\end{scope}

\begin{scope}[every node/.style={circle}]
    \node (A) at (1,1.5) {$\mathrm{A}$};
    \node (B) at (2,1.5) {$\mathrm{B}$};
    \node (C) at (3,1.5) {$\mathrm{C}$};
    \node (D) at (4,1.5) {$\mathrm{D}$};
    
    \node (E) at (5.25,1.5) {$\mathrm{C}$};
    \node (F) at (6.25,1.5) {$\mathrm{A}$};
    \node (G) at (7.25,1.5) {$\mathrm{D}$};
    \node (H) at (8.25,1.5) {$\mathrm{B}$};
\end{scope}
\begin{scope}[every node/.style={fill=white,circle},
              every edge/.style={draw=black,very thick}]
    \path [-] (A) edge[dotted] (B);
    \path [-] (B) edge[dotted] (C);
    \path [-] (C) edge[dotted] (D);
    
    \path [-] (E) edge (F);
    \path [-] (F) edge (G);
    \path [-] (G) edge (H);
    \path [-](D) edge[double, -{Implies}] (E);
\end{scope}

\begin{scope}[every node/.style={circle}]
    \node (A) at (0.5,0.5) {};
    \node[right] (B) at (1.5,0.5) {Separable};
    \node (C) at (4.5,0.5) {};
    \node[right] (D) at (5.5, 0.5) {Entangled};
\end{scope}
\begin{scope}[every node/.style={fill=white,circle},
              every edge/.style={draw=black,very thick}]
    \path [-] (A) edge[dotted] (B);
    \path [-] (C) edge (D);
\end{scope}
\draw (0,1) -- (8.5,1);

\end{tikzpicture}

    \caption{A schematic diagram for the metatransitivity examples. Each row describes the known bipartite marginals in a 3- or 4-partite system and the target subsystems where metatransitivity are exhibited. }
    \label{fig:MetaEx}
\end{figure}

\subsection*{Metatransitivity from marginals of random pure states}

Naturally, one may wonder how common the phenomena of (meta)transitivity is. Our numerical results based on pure states randomly generated according to the Haar measure suggest that transitivity {\em is} generic in the tripartite scenario: for local dimension up to five, {\em all} sampled pure states have only non-PPT marginals and demonstrate entanglement transitivity. However, with more subsystems,  (meta)transitivity seems rare. For example, among the $10^5$  sampled four-qubit states, only about $7.32\%$ show transitivity while about $3.38\%$ show metatransitivity. For a system with even more subsystems or with a higher $d$, we do not find any example of (meta)transitivity from random sampling (see Table~\ref{Tbl.Numerics}).

Next, notice that for the convenience of verification, some explicit examples that we provide actually involve marginals leading to a {\em unique} global state. However, uniqueness is not {\em a priori} required for entanglement (meta)transitivity. For example, among those quadripartite (meta)transitivity examples found for randomly sampled pure states, $>73\%$ of them (see Supplementary Note~5) are {\em not} uniquely determined from three of its two-qubit marginals (cf. Ref.~\cite{Jones.PRA.2005,Huber.PRL.2016,wyderkahuberguhne2017}). In contrast, most of the tripartite numerical examples found appear to be uniquely determined by two of their two-qudit marginals, a fact that may be of independent interest (see, e.g., Ref.~\cite{lindenpopescuwootters2002,Linden.PRL.2002b,Diosi.PRA.2004,Han.PRA.2005}).

\section*{Discussion} 
\label{sec.discussion}

The example involving noisy $W$-state marginals demonstrate that the transitivity can occur for arbitarily long chain of quantum systems.
This leads us to consider metatransitivity with only separable marginals. Beyond the example given above, we present also in Supplementary Note~5 a five-qubit example with four separable marginals and discuss some possibility to extend the chain. For future work, it could be interesting to determine if such exotic metatransitivity examples exist at the two ends of an arbitrarily long chain of multipartite system. For the closely related EMP, we remind that an explicit construction for a state with only two-body separable marginals and an arbitrarily large number of subsystems is known~\cite{Paraschiv2018} (see also Ref.~\cite{NavascuesBaccariAcin2021}).

So far, we have discussed only cases where both the input marginals and the target marginal are for two-body subsystems. If entanglement can be deduced from two-body marginals, it is also deducible from higher-order marginals that include the former from coarse graining. Hence, the consideration of two-body input marginals allows us to focus on the crux of the ETP. As for the target system, we provide---as an illustration---in Supplementary Note~5 an example where the three two-qubit marginals of \cref{fig.TreeGraph}(b) imply the genuine three-qubit entanglement present in BCD. Evidently, there are many other possibilities to be considered in the future, as entanglement in a multipartite setting is known~\cite{horodeckicubed2009,Guhne:2009wn} to be far richer.

Our metatransitivity examples also illustrate the disparity between the local compatibility of probability distributions and quantum states. Classically, probability distributions $P(A,B)$ and $P(B,C)$  compatible in $P(B)$ always have a joint distribution $P(A,B,C)$ (this extends to the multipartite case for marginal distributions that form a tree graph~\cite{wolfVerstraeteCirac2003}). One may think that the quantum analogue of this is: compatible $\rho_{\rab}$ and $\rho_{\rbc}$ must imply a separable joint state, and hence a separable $\rho_{\rac}$. However, our metatransitivity example (as with nontrivial instances of tripartite EMPs), illustrates that this generalization does not hold. Rather, as we show in Supplementary Note~8, a possible generalization is given by classical-quantum states $\rho_{\rab}$ and $\rho_{\rbc}$ sharing the same diagonal state in $\rb$ --- in this case, metatransitivity can never be established.

Evidently, there are many other possible research directions that one may take from here. For example, as with the $W$-states, we have also observed transitivity in $n \le 3 \le d \le 6$ for qudit Dicke states~\cite{Dicke1954,WeiGoldbart2003,Aloy_2021}, which seems to be also uniquely determined by its $(n-1)$ bipartite marginals. To our knowledge, this uniqueness remains an open problem and, if proven, may allow us to establish examples of transitivity for an arbitrarily high-dimensional quantum state that involves an arbitrary number of particles. From an experimental viewpoint, the construction of witnesses specifically catered for ETPs are surely welcome.

Finally, notice that while ETPs include EMPs as a special case, an ETP may be seen as an instance of the more general resource transitivity problem~\cite{RTP}, where one wishes to certify the resourceful nature of some subsystem based on the information of {\em other} subsystems. In turn, the latter can be seen as a special case of the even more general resource marginal problems~\cite{RMP2022}, where resource theories are naturally incorporated with the marginal problems of quantum states.

\section*{Methods}

\subsection*{Metatransitivity certified using separability criteria}

As mentioned before, we can certify the entanglement (meta)transitivity of a given set of marginals in a bipartite target system $\rmt$ by demonstrating the violation of the PPT separability criterion. We can show this by solving the following convex optimization problem:
\begin{align}
\label{eq.NPTtransitivitySDP}
\nonumber
&\max_{\rhos}\qquad\qquad\qquad \lambda \\
&\text{subj. to }  \quad
 \tr_{\rms\backslash \rms_{i}}(\rhos) = \sigma_{\rms_{i}}\,\, \forall\,\, \rms_i\in \mathcal{S},\nonumber\\ 
 &\qquad\qquad\,\, \rhos \succeq 0, \quad\rhot^{\Gamma} \succeq  \lambda \mathbb{I},
\end{align}
which directly optimizes over the joint state $\rhos$ with marginals $\sigma_{\rms_i}$ such that the smallest eigenvalue $\lambda$ of $\rhot^{\Gamma}$ is maximized. Because a bipartite state that is {\em not} PPT is entangled~\cite{peres1996,horodecki1996}, if the optimal $\lambda$ (denoted by $\lambda^\star$ throughout) is negative, the marginal state in $\rmt$ of all possible joint states $\rhos$  must be entangled.

In the Supplementary Notes, we compute the Lagrange dual problem to~\cref{eq.transitivityOpt} with a linear witness $W_{T}$. A similar calculation for~\cref{eq.NPTtransitivitySDP} 
shows that it is equivalent to a dual problem with $W=\eta_\rmt^\Gamma$, where $\eta_\rmt$ being an additional optimization variable subjected to the constraint of $\eta_\rmt\succeq 0$ and $\tr(\eta_\rmt)=1$.

Meanwhile, to certify genuine tripartite entanglement in the target tripartite marginal $\rmt$, we use a simple criterion introduced in~\cite{Li2017}. 
Consider the density operator $\rho_{AB}$ on $\mathbb{C}^m \otimes \mathbb{C}^n$ to be an $m\times m$ block matrix of $n\times n$ matrices $\rho^{(i,j)}$. Let $\widetilde{\rho_{AB}}$ denote the realigned matrix obtained by transforming each block $\rho^{(i,j)}$ into rows. 
The CCNR criterion~\cite{Rudolph2005,ChenWu2003} dictates that for separable $\sigma_{\rab}$, $\lVert \widetilde{\sigma_{\rab}} \rVert_{1} \leq 1$.

Now, let ${\rm A|BC}$ denote a bipartition of a tripartite system $\rabc$ into a bipartite system with parts $\rm A$ and $\rbc$. 
Finally, for any tripartite state $\rho_{\rabc}$ on  $\mathbb{C}^d\otimes \mathbb{C}^d\otimes \mathbb{C}^d$,
define
\begin{align}
\nonumber
    M(\rho_{\rabc}) &:= \tfrac{1}{3}\left( \lVert \rho_{\rabc}^{\text{\tiny T$_\ra$}} \rVert_1 +
    \lVert \rho_{\rabc}^{\text{\tiny T$_\rb$}} \rVert_1 + \lVert \rho_{\rabc}^{\text{\tiny T$_\rc$}} \rVert_1 \right) \\
    N(\rho_{\rabc}) &:= \tfrac{1}{3}\left( \lVert \widetilde{\rho_{\rm A|BC}} \rVert_1 +
    \lVert \widetilde{\rho_{\rm B|CA}} \rVert_1 + \lVert \widetilde{\rho_{\rm C|AB}} \rVert_1 \right),
    \label{Eq.MandN}
\end{align}
where $\text{\tiny T$_{X}$}$ means a partial transposition with respect to the subsystem ${\rm X}$.
It was shown in~\cite{Li2017} that for any biseparable $\rho_\rabc$, we must have
 \begin{equation}
     \max\{M(\rho_{\rabc}), N(\rho_{\rabc}) \} \le \tfrac{1+2d}{3}.
 \end{equation}
 This means that if any of $M(\rho_{\rabc}), N(\rho_{\rabc})$ is larger than $\tfrac{1+2d}{3}$,
 $\rho_{\rabc}$ must be genuinely tripartite entangled. 
 
 Therefore in the metatransitivity problem, we can use this, cf. \cref{eq.NPTtransitivitySDP} for the bipartite target system, for detecting genuine tripartite entanglement.
 This is done by minimizing $M$ and $N$ of the target marginal and taking the larger of the two minima. 
To this end, note that the minimization of the trace norm can be cast as an SDP~\cite{BenTalNemirovskii2001}. 
  Further details can be found in Supplementary Notes 1.
 
 \subsection*{Certifying the uniqueness of a global compatible (pure) state}\label{app.certUniqueGlobal}

A handy way of certifying the (meta)transitivity of marginals $\{\sigma_{\rms_i}\}$ known to be compatible with some pure state $\ket{\psi}$ is to show that the global state $\rhos$ compatible with these marginals is {\em unique}, i.e., $\rhos$ is necessarily $\proj{\psi}$. This can be achieved by solving the following SDP:
\begin{align}
\nonumber
&\min_{\rhos}\qquad\qquad\qquad \bra{\psi}\rhos\ket{\psi} \\
&\text{subj. to }  \quad
 \tr_{\rms\backslash \rms_{i}}(\rhos) = \sigma_{\rms_{i}}\,\, \forall\,\, \rms_i\in \mathcal{S} \text{ and } \rhos  \succeq 0
\end{align}
The objective function here is the fidelity of $\rhos$ with respect to the pure state $\ket{\psi}$. If this minimum is $1$, then by the property of the Uhlmann-Jozsa fidelity~\cite{Liang:2019wb}, we know that the only compatible $\rhos$ is indeed given by $\proj{\psi}$.

For the numerical results that show how typical transitivity is for the bipartite marginals of a pure global state, the marginals are obtained from a uniform random $n$-qudit state, which is obtained by taking the first column of a $d^n$-dimensional Haar-random unitary.

\section*{Data Availability}

All relevant data supporting the main conclusions and figures of the document are available on request. Please refer to Gelo Noel Tabia at 
\href{mailto:gelonoel-tabia@gs.ncku.edu.tw}{gelonoel-tabia@gs.ncku.edu.tw}.

\section*{Acknowledgments}
We thank Antonio Ac\'{\i}n and Otfried G\"uhne for helpful comments. 
CYH is supported by ICFOstepstone (the Marie Sk\l odowska-Curie Co-fund GA665884), the Spanish MINECO (Severo Ochoa SEV-2015-0522), the Government of Spain (FIS2020-TRANQI and Severo Ochoa CEX2019-000910-S), Fundaci\'o Cellex, Fundaci\'o Mir-Puig, Generalitat de Catalunya (SGR1381 and CERCA Programme), the ERC AdG CERQUTE, and the AXA Chair in Quantum Information Science.
We also acknowledge support from 
the Ministry of Science and Technology, Taiwan (Grants No. 107-2112-M-006-005-MY2, 109-2627-M-006-004, 109-2112-M-006-010-MY3), and the National Center for Theoretical Sciences, Taiwan.

\section*{AUTHOR CONTRIBUTIONS}

GNT, CYH, and YCL formulated the problem and developed the theoretical ideas.
GNT, YCL, and YCY carried out the numerical calculations.
All explicit examples are due to GNT and KSC.
GNT and YCL prepared the manuscript with help from KSC and CYH.
All authors were involved in the discussion and interpretation of results.



\clearpage
\onecolumngrid

\begin{table}[h!]
\setlength\tabcolsep{2pt} %


\begin{tabular}{c||c|cc|ccc|cc|c}
\hline \hline
$(n,d)$ & $N_{\text{\tiny sample}}$  & NPT  & PPT & NPT   & PPT  &  NPT + PPT   & Max $(1-\mathcal{F})$ & $1-\mathcal{F}<\epsilon$ & $1-\mathcal{F}<10^{-6}$  \\ 
 & ($\times10^3$)  & (\%) &  (\%) & $\Rightarrow$ NPT (\%)  &  $\Rightarrow$ NPT (\%) &  $\Rightarrow$ NPT (\%)  &  &  (\%)  & among $\Rightarrow$ NPT (\%) \\ \hline
(3,2) & $1000$  & 100 &  0 & 100 & 0 (-) & 0 (-) & $1.21\times10^{-9}$ & 100; 100; 100 & 100 \\
(3,3) & $100$  & 100 & 0 & 100 & 0 (-) & 0 (-)  & $1.73\times10^{-6}$ & 25.51; 69.55; 99.99 & 99.99\\
(3,4) & $10$  & 100 & 0 & 100 & 0 (-) & 0 (-)  & $1.33\times10^{-6}$ & 72.77; 85.47; 99.84 & 99.84\\
(3,5) & $10$  & 100 & 0 & 100 & 0 (-)& 0 (-) & $1.29\times10^{-6}$ & 83.64; 94.31; 99.79 & 99.79\\
(4,2) & $100$   & 46.74 & 2.64 & 7.32 (15.66)  & 0.02 (0.87) & 3.36 (6.64)  &  $1^\dag$ & 0.29; 1.50; 3.20 & 26.18 \\
(4,3) & $10$ & 99.93 & 0 & 0 (0)  & 0 (-) & 0 (0) &  $1^\dag$  & 0.00; 0.00; 0.00 & -\\
(5,2) & $10$   & 0.35 & 45.75 & 0 (0)  & 0 (0) & 0 (0) &  $1^\dag$  & 0.00; 0.00; 0.00 & -\\
(5,3) & $1.030$  & 0.10 & 54.85 & 0 (0) & 0 (0) & 0 (0) & $1^\dag$ & 0.00; 0.00; 0.00 & -
\\
\hline \hline 
\end{tabular}

\centering 
\caption{\label{Tbl.Numerics} Summary of various features of uniformly sampled $n$-partite pure states of local dimension $d$ according to the Haar measure. The second column gives the number of pure states sampled $N_{\text{\tiny sample}}$ in each scenario $(n,d)$. The next two columns list the fraction of states giving $(n-1)$ neighboring two-body marginals that are, respectively,  {\em all} NPT (i.e., none of which being PPT) and {\em all} PPT. The next three columns summarize how generic the phenomenon of (meta)transivitiy is among such states when the target system $\rmt$ lie at the two ends of an $n$-body chain. 
We give from left to right, respectively, the fraction among all sampled states exhibiting transitivity (i.e., with {\em only} entangled marginals), metatransitivity with {\em only} separable marginals, and metatransitivity with mixed marginals. Enclosed in each bracket is the corresponding fraction among samples having the associated kind of marginals.
The next two columns summarize the extent to which the $(n-1)$ two-body marginals lead to a unique global pure state. These are expressed in terms of the {\em largest} value of the infidelity $1-\mathcal{F}$, where $\mathcal{F}=\min_{\rhos} \bra{\psi}\rhos\ket{\psi}$ and $\ket{\psi}$ is the sampled pure state; the three numbers listed in the second last column are, respectively, for $\epsilon=10^{-8},10^{-7}$ and $10^{-6}$. The final column shows the fraction of (meta)transitivity examples having a {\em unique} global state (with an infidelity threshold set to $10^{-6}$). Throughout, we use $1^\dag$ to represent a number that differs from $1$ by less than $10^{-8}$.
}
\end{table}

\twocolumngrid

\bibliographystyle{apsrev4-1}

\input{NPJ-QI-ETP_Arxiv.bbl}
\input{NPJ-QI-ETP_Arxiv_SuppInfo.tex}


\end{document}

%% file: NPJ-QI-ETP_Arxiv_SuppInfo.tex
\clearpage
\newpage
\setcounter{figure}{5}

\setcounter{secnumdepth}{2}
\makeatletter
\renewcommand*{\thesection}{Supplementary Note ~\arabic{section}}
\renewcommand*{\thesubsection}{\arabic{section}.\arabic{subsection}}
\renewcommand*{\p@subsection}{}
\renewcommand*{\thesubsubsection}{\thesubsection.\arabic{subsubsection}}
\renewcommand*{\p@subsubsection}{}
\makeatother

\section{Various optimization problems}
\label{app.VariousOptimization}
\subsection{Certification of entanglement (meta)transitivity via a linear witness}\label{app.certLinearWitness}

\subsubsection{Lagrange dual problem to Eq.~(1)} 
\label{App.DualProblem}

For the optimization problem in 
Eq.~(1)
where
\begin{equation}
\label{eq.linearWitness}
    \mathcal{W}(\rhot) = \tr[\rhos(W_\rmt\otimes\mathbb{I}_{\rms\backslash \rmt})]
\end{equation}
and $W_\rmt$ is some Hermitian operator, we can construct the Lagrangian~\cite{boyd2004convex}
\begin{align*}
    \mathcal{L}(\rhos, H_{\rms_i},Z)
    = \innerprod{W_\rmt}{\rho_\rmt}
    - \sum_{i=1}^{k} \innerprod{H_{\rms_i}}{\rho_{\rms_i}-\sigma_{\rms_i}} + \innerprod{Z}{\rhos},
\end{align*}
where the Lagrange multipliers $H_{\rms_i}$ are Hermitian and $Z \succeq 0$. For convenience, let 
\begin{equation}\label{Eq.Dfn.zeta}
	\zeta_\rms:=\sum_{i} H_{\rms_i}\otimes\mathbb{I}_{\rms\backslash \rms_i}-W_\rmt \otimes \mathbb{I}_{\rms \backslash \rmt}, 
\end{equation}	
the dual function~\cite{boyd2004convex} $g(H_{\rms_i},Z):= \sup_{\rhos } \mathcal{L}(\rhos, H_{\rms_i}, Z)$ is 
\begin{align}
    g(H_{\rms_i},Z) = \sup_{\rhos} 
    \innerprod{Z-\zeta_\rms}{\rhos} + \sum_{i=1}^{k} \innerprod{H_{\rms_i}}{\sigma_{\rms_i}}. 
\end{align}
Thus, unless  $Z =\zeta_\rms$, the dual function becomes unbounded, i.e., $g(H_{\rms_i},Z) = +\infty$ by choosing $\rhos$ to be an eigenstate of $Z-\zeta_\rms$ with non-vanishing eigenvalue and by making the norm of that eigenstate arbitrarily large.
Incorporating the non-negativity of $Z$ and eliminating it from the problem then gives the Lagrange dual problem
\begin{align}
\nonumber
& \min_{\{H_{\rms_i}\}_{S_{i} \in \mathcal{S}}}\qquad   \sum_{i=1}^{k} \tr\left(\sigma_{\rms_i} H_{\rms_i} \right) \\
& \text{subj. to } \quad H_{\rms_i} =H_{\rms_i}^\dag\,\, \forall\,\, \rms_i\in \mathcal{S} \text{ and } \zeta_\rms  \succeq 0.
\label{Eq.Dual.general}
\end{align}

\subsubsection{Manifestation of (meta)transitivity by \texorpdfstring{\cref{Eq.Dual.general}}{Eq.~(15)}}

Here it will be convenient to follow Proposition 1.19 on page 55 of Ref.~\cite{Watrous2018}. 
For this,  we will need to write the primal semidefinite program (SDP) in the form
\begin{gather}
\label{eq.watprimal}
\max_{X}  \innerprod{A}{X}   \,\,
\text{subj. to } \,\, \Phi(X) = B \text{ and } 
X \succeq 0,
\end{gather}
where $\innerprod{M_1}{M_2} = \tr(M_1^\dag M_2)$.
This means the dual SDP can be expressed as
\begin{gather}
\label{eq.watdual}
\min_{Y}  \innerprod{B}{Y}  \,\,
\text{subj. to } \,\, \Phi^\dag(Y) \succeq A \text{ and }  Y = Y^\dag.
\end{gather}
When strong duality holds, i.e., when the primal value $\innerprod{A}{X}$ coincides with the dual value $\innerprod{B}{Y}$ for some $X$ and $Y$, {\em complementary slackness} dictates that~\cite{Watrous2018} (see also Ref.~\cite{boyd2004convex})
\begin{equation}
\label{eq.compSlack}
    [\Phi^\dag(Y)- A]X = 0.
\end{equation}

It is straightforward to verify that Eq.~(1)
with linear witness given by \cref{eq.linearWitness} can be written in the form of \cref{eq.watprimal} by taking $A = W_\rmt\otimes \id_{\rms \backslash \rmt},X = \rho_\rms$ with
\begin{align*}
    B &= \bigoplus_{i=1}^k \sigma_{\rms_{i}}, &
    \Phi(X) &=\bigoplus_{i=1}^{k} \tr_{\rms \backslash \rms_{i}}(X).
\end{align*}
Similarly, \cref{Eq.Dual.general} is in the form of \cref{eq.watdual} by setting
\begin{align*}
    Y = \bigoplus_{i=1}^k H_{\rms_{i}},\ \ 
    \Phi^\dag(Y) =
    \sum_{i=1}^{k} H_{\rms_{i}} \otimes \id_{\rms \backslash \rms_{i}},\ \  A = W_\rmt\otimes \id_{\rms \backslash \rmt}
\end{align*}
where we used the fact that the adjoint channel of partial trace is tensoring by identity. Finally,
from \cref{eq.compSlack} we have that if strong duality holds then
\begin{equation}
    \left(  \sum_{i=1}^{k} H_{\rms_{i}}^\star \otimes \id_{\rms \backslash \rms_{i}} - W_\rmt\otimes \id_{\rms \backslash \rmt} \right) \rho_{\rms}^\star = \zeta_\rms^\star \rho_{\rms}^\star = 0 
\end{equation}
for the optimal joint state $\rho_\rms^\star$ and optimal dual variables $H_{\rms_i}^\star$. The last equality, in particular, implies that the pair $(\zeta_\rms^\star, \rho_{\rms}^\star)$ satisfies $\tr(\rho_{\rms}^\star\zeta_\rms^\star)=0$. Since $\zeta_{\rms} \succeq 0$, this last equality further implies that whenever strong duality holds, $\rho_{\rms}^\star$ must be a (mixture) of ground states of the Hamiltonian $\zeta_{\rms}$. 

Finally, when metatransitivity is certified by the witness $\W$, i.e., $\W(\rhot)<0$, the local interaction energy at $\rmt$ must satisfy
\begin{equation*}
	E_{\rmt} = -\tr(\rhos\,W_\rmt\otimes\id_{\rms\setminus\rmt}) = -\tr(\rhot\,W_\rmt)= -\W(\rhot)>0.
\end{equation*}

\subsection{SDPs for certifying entanglement (meta)transitivity via a violation of some separability criterion}\label{app.SeparabilityCriterion}

\subsubsection{The PPT separability criterion}

As mentioned in the main text, if we take 
in Eq.~(1) $W_\rmt = \eta_\rmt^\Gamma$
with $\eta_\rmt\succeq 0$, where $\Gamma$ denotes the partial transposition operation, and optimize over all such $\eta$, we end up with a witness that allows us to certify the entanglement transitivity via a violation of the PPT separability criterion. Such an optimization is, however, bilinear in $\rhos$ and $\eta$, and thus does not fit into the framework of a convex optimization problem.

To circumvent this problem, one can make use of the following optimization problem:
\begin{align}
\label{eq.transitivitySDP}
\nonumber
&\max_{\rhos}\qquad\qquad\qquad \lambda \\
&\text{subj. to }  \quad
 \tr_{\rms\backslash \rms_{i}}(\rhos) = \sigma_{\rms_{i}}\,\, \forall\,\, \rms_i\in \mathcal{S},\nonumber\\ 
 &\qquad\qquad\,\, \rhos \succeq 0, \quad\rhot^{\Gamma} \succeq  \lambda \mathbb{I},
\end{align}
which directly optimizes over the joint state $\rhos$ with marginals $\sigma_{\rms_i}$ such that the smallest eigenvalue $\lambda$ of $\rhot^{\Gamma}$ is maximized.  Since a bipartite state that is {\em not} PPT is entangled~\cite{peres1996,horodecki1996}, if the optimal $\lambda$ (denoted by $\lambda^\star$ throughout) is negative, the marginal state in $\rmt$ of all possible joint states $\rhos$  must be entangled. By following a calculation similar to the one given above, one can show that the Lagrange dual problem to~\cref{eq.transitivitySDP} takes exactly the same form as~\cref{Eq.Dual.general}, but with $W=\eta_\rmt^\Gamma$, and with $\eta_\rmt$ being an additional optimization variable subjected to the constraint of $\eta_\rmt\succeq 0$ and $\tr(\eta_\rmt)=1$.

\subsubsection{Some other means of certifying entanglement transitivity}

To certify the entanglement in $\rmt$, we may use different kinds of entanglement detection criteria. For example, if we employ the so-called ESIC criterion based on symmetric informationally complete positive operator-valued measures (SIC-POVMs)~\cite{Shang2018}, which is similar to the computable cross-norm or realignment (CCNR) criterion~\cite{Rudolph2005,ChenWu2003,Horodecki2006} except that each set of local orthogonal observables is replaced by a single SIC-POVM. Then for the target system $\rmt=(t,t')$ we can instead compute
\begin{align}
\label{eq.criterionSIC}
\nonumber
\min_{\rhos} \norm{\mathcal{P}^\rmt}_{1}, &\,
\text{subject to }  \,
\, \tr_{\rms\backslash \rms_{i}}(\rhos) = \sigma_{\rms_{i}} \forall \rms_{i}\in\mathcal{S}, \\
&\, \rhos \ge 0, \,
\mathcal{P}_{ij}^{\rmt} = \tr\left(\rhot E_{i}^{t}\otimes E_{j}^{t'}\right), \forall i,j,
\end{align}
where $\norm{M}_{1} :=\tr\sqrt{M^\dag M}$ is the trace norm (i.e., the sum of the singular values) of $M$, and the operators 
$E_{i}^{t} = \sqrt{\tfrac{d+1}{2d}} \proj{\psi_{i}^{t}}$ 
are constructed from the set $\{\ket{\psi_i^{t}}: i=1,\ldots,d^2\}$ whose projectors correspond to a SIC-POVM~\cite{Zauner2011,Renes2004,ScottGrassl2010,FuchsHoangStacey2017}. For this criterion, we can certify the entanglement in $\rmt$ when the optimal
$\norm{\mathcal{P}}_{1} > 1$, which is independent of the chosen $\{ E_{j}^{t} \}$ for each target subsystem~\cite{Shang2018}. 

\subsubsection{The PPT and CCNR criterion for genuine tripartite entanglement}
\label{App.GME.Criterion}

Here, we explain a simple criterion for detecting genuine tripartite entanglement introduced in~\cite{Li2017}.
To this end, we first briefly recall from~\cite{ChenWu2003} the realignment operation, which is based upon $\mathrm{vec}(\mathcal{M})$, the operation of rearranging the columns of the matrix $\mathcal{M}$ into a column vector (i.e., for standard basis vectors $\ket{i}$, $\mathrm{vec}(\ketbra{i}{j}) = \ket{j}\ket{i}$). 

Given a bipartite density operator $\rho_{\rab}$ acting on $\mathbb{C}^m\otimes \mathbb{C}^n$, we may write it as an $m\times m$ block matrix, 
\begin{equation}
    \rho_{AB} = 
    \begin{pmatrix}
    \rho^{(11)} &  \rho^{(12)}  & \cdots & \rho^{1m}\\
    \rho^{(21)} &  \rho^{(22)}  & \cdots & \rho^{2m} \\
    \vdots      &  \vdots       & \ddots & \vdots \\
    \rho^{(m1)} & \cdots        &  \cdots & \rho^{(mm)}
    \end{pmatrix},
\end{equation}
where each $\rho^{(ij)}$ is an $n\times n$ matrix. Then, we can construct a $m^2 \times n^2$ realigned matrix $\widetilde{\rho_{\rab}}$
\begin{equation}
    \widetilde{\rho_{\rab}} =
    \begin{pmatrix}
    \mathrm{vec}\left( \rho^{(11)} \right)^{\text{\tiny\rm T}} \\
    \mathrm{vec}\left( \rho^{(21)} \right)^{\text{\tiny\rm T}} \\
    \vdots \\
    \mathrm{vec}\left( \rho^{(mm)} \right)^{\text{\tiny\rm T}} \\
    \end{pmatrix}.
\end{equation}
In other words, the realigned matrix is obtained by turning the $m \times m$ blocks into rows. The CCNR criterion~\cite{Rudolph2005,ChenWu2003} dictates that for separable $\sigma_{\rab}$, $\lVert \widetilde{\sigma_{\rab}} \rVert_{1} \leq 1$.

Now, let ${\rm A|BC}$ denote a bipartition of a tripartite system $\rabc$ into a bipartite system with parts $\rm A$ and $\rbc$. Then, a biseparable state $\varrho_{\text{\tiny\rm ABC}}^{\rm bs.}$ is a convex mixture of states separable with respect to the different bipartitions, i.e., 
\begin{align}
    \varrho_{\text{\tiny\rm ABC}}^{\rm bs.} \!&=\! \sum_i\alpha_i \varrho^{(i)}_{\text{\tiny\ra}}\otimes\varrho^{(i)}_{\text{\tiny\rbc}}\!\! +\!\! \sum_j\beta_j \varrho^{(j)}_{\text{\tiny\rb}}\otimes\varrho^{(j)}_{\text{\tiny\rm CA}}\!\! +\!\! \sum_k\gamma_k \varrho^{(k)}_{\text{\tiny\rc}}\otimes\varrho^{(k)}_{\text{\tiny\rab}},\nonumber\\
   & \alpha_i,\beta_j,\gamma_k\ge 0,\quad \sum_i\alpha_i+\sum_j\beta_j+\sum_k\gamma_k=1,
   \label{Eq.bisep.}
\end{align}
where $\varrho^{(i)}_{\text{\tiny\ra}},\varrho^{(i)}_{\text{\tiny\rbc}},\varrho^{(j)}_{\text{\tiny\rb}},\varrho^{(j)}_{\text{\tiny\rm CA}},\varrho^{(k)}_{\text{\tiny\rc}},\varrho^{(k)}_{\text{\tiny\rab}}$ are normalized density matrices.

Furthermore, let $\rho_{\rabc}$ be a three-qudit density operator acting on $\mathbb{C}^d\otimes \mathbb{C}^d\otimes \mathbb{C}^d$ and  
\begin{align}
\nonumber
    M(\rho_{\rabc}) &= \tfrac{1}{3}\left( \lVert \rho_{\rabc}^{\text{\tiny T$_\ra$}} \rVert_1 +
    \lVert \rho_{\rabc}^{\text{\tiny T$_\rb$}} \rVert_1 + \lVert \rho_{\rabc}^{\text{\tiny T$_\rc$}} \rVert_1 \right) \\
    N(\rho_{\rabc}) &= \tfrac{1}{3}\left( \lVert \widetilde{\rho_{\rm A|BC}} \rVert_1 +
    \lVert \widetilde{\rho_{\rm B|CA}} \rVert_1 + \lVert \widetilde{\rho_{\rm C|AB}} \rVert_1 \right),
    \label{Eq.defMandN}
\end{align}
where $\text{\tiny T$_{X}$}$ means a partial transposition with respect to the subsystem ${\rm X}$.
In these notations, it was shown~\cite{Li2017} that for any biseparable $\rho_\rabc$, cf. \cref{Eq.bisep.}, we must have
 \begin{equation}
     \max\{M(\rho_{\rabc}), N(\rho_{\rabc}) \} \le \tfrac{1+2d}{3}.
 \end{equation}
 This means that if any of $M(\rho_{\rabc}), N(\rho_{\rabc})$ is larger than $\tfrac{1+2d}{3}$,
 $\rho_{\rabc}$ must be genuinely tripartite entangled.
 
 Therefore in the metatransitivity problem, we can use this, cf. \cref{eq.transitivitySDP} for the bipartite target system, for detecting genuine tripartite entanglement.
 This is done by minimizing $M$ and $N$ of the target marginal and taking the larger of the two minima. 
To this end, note that the minimization of the trace norm can be cast as an SDP~\cite{BenTalNemirovskii2001}. One approach is to recognize that the singular values of a matrix $M$ can be obtained from the nonzero eigenvalues of the symmetric matrix 
 $\Omega = \begin{pmatrix}
 0 & M \\ M^\dag & 0
 \end{pmatrix}$. More precisely, if $M$ has singular values $m_{i}$ then $\Omega$
 will have nonzero eigenvalues $\omega_i =\pm m_{i}$. This means that minimizing the trace norm of $M$ is equivalent to minimizing half of the $\ell_{1}$-norm
 $\lVert \vec{\omega} \rVert_{1} = \sum_i |\omega_i|$ of the vector of eigenvalues $\vec{\omega}$ of $\Omega$. This in turn can be solved by the SDP
 \begin{gather}
 \nonumber
     \min_{\Omega^+,\Omega^-} \tfrac{1}{2}(\tr(\Omega^+) + \tr(\Omega^-)), \,\,
     \text{subj. to } \,\,
     \Omega = \Omega^{+} - \Omega^{-}, \\
     \Omega^{+} \ge 0 \text{ and } \Omega^{-} \ge 0.
 \end{gather}



\section{A family of \texorpdfstring{$n$}{n}-qubit states exhibiting transitivity}\label{app.NqubitTransitivity}

For all integers $n\ge 3$, consider the $n$-qubit mixed state:
\begin{equation}
	\Omega_{n}(\gamma) = \gamma\proj{W_n} + (1-\gamma)\proj{0^n},\quad \gamma\in(0,1],
\end{equation}
which is a mixture of $\proj{0^n}$ and the $n$-qubit $W$ state. It is straightforward to verify that its two-qubit reduced states are:
\begin{equation}\label{eq.Rho_nGamma}
    \rho_{n}(\gamma) = \left(\tfrac{n-2\gamma}{n} \right)
    \proj{00} + \tfrac{2\gamma}{n}\proj{\Psi^+},
\end{equation}
where $\ket{\Psi^+}=\tfrac{1}{\sqrt{2}}(\ket{10} +\ket{01})$. In what follows, we show that for any $n$-vertex tree graph whose edges correspond to the bipartite marginals $\rho_n(\gamma)$, i.e., 
\begin{equation}\label{Eq.Wstate_marginal}
    \tr_{\rms\backslash \rms_i} (\rho) = \sigma_{\rms_i}=\rho_n(\gamma) \,\,\forall\,\, \rms_i \in \calS,
\end{equation}
 the global state $\rho$ compatible with these marginals in tree form is {\em unique} and hence given by $\Omega_{n}(\gamma)$. We begin by proving a lemma pertaining to the structure of the eigenstates of $\rho$.

\begin{lemma}
\label{lem.if_01_then_10}
Let $\rms$ be the global system, $\rms_i \in \calS$ be any two-qubit subsystem with marginal specified as $\rho_{n}(\gamma)$, and
\begin{equation}\label{Eq.Expand}
	\ket{\Psi_\ell} = \sum_{i_1,i_2,\cdots,i_n=0,1} \alpha^{(\ell)}_{i_1,i_2,\cdots,i_n}\ket{i_1i_2\cdots i_n},
\end{equation}
be an eigenstate of $\rho$ with nonzero eigenvalue, then \\
(i) all amplitudes $\alpha^{(\ell)}_{i_1,i_2,\cdots,i_n}$ with  two ``$1$" at the positions of $\rms_i$ vanish;\\
(ii) the amplitudes $\alpha^{(\ell)}_{i_1,i_2,\cdots,i_n}$ with one ``$1$" and one ``$0$" at the positions of $\rms_i$ are identical.\\
\end{lemma}
\begin{proof}
Let us write the global state $\rho$ in its spectral decomposition:
\begin{equation}\label{Eq:Spectral_Decomposition_of_state}
    \rho = \sum_\ell c_\ell \proj{\Psi_\ell},
\end{equation}
where $\braket{\Psi_i}{\Psi_j}=0~\forall\,\, i\neq j$, $\sum_i c_i = 1$, and $c_i >0~\forall\,\, i$ are the nonzero eigenvalues of $\rho$.

Without loss of generality, let $\rms_i$ be the first two qubits (otherwise, reorder the particles to make them so), then 
\begin{align*}
	0 = &\,\,\bra{11}\rho_n(\gamma)\ket{11}= \bra{11}\tr_{\rms\setminus\rms_i}(\rho)\ket{11}\\
	= &\sum_\ell c_\ell \bra{11}\tr_{\rms\setminus\rms_i}(\proj{\Psi_\ell})\ket{11}\\
	= &\sum_\ell c_\ell\, \tr_{\rms\setminus\rms_i}\left[\left(\bra{11}_{\rms_i}\otimes\id_{\rms\setminus\rms_i}\right)\proj{\Psi_\ell}\left(\ket{11}_{\rms_i}\otimes\id_{\rms\setminus\rms_i}\right)\right]\\
	= &\sum_\ell c_\ell |\alpha^{(\ell)}_{1,1,i_3,\cdots,i_n}|^2
\end{align*}
where the first equality follows from \cref{eq.Rho_nGamma}, second equality follows from \cref{Eq.Wstate_marginal}, third equality follows from \cref{Eq:Spectral_Decomposition_of_state}, and the last equality follows from \cref{Eq.Expand}.
Since the last expression is a convex sum of non-negative terms, the fact that the sum vanishes means that each $\alpha^{(\ell)}_{1,1,i_3,\cdots,i_n}$ is zero 
for all $\ell,i_3,i_4,\cdots,i_n$ and $\rms_i$ as claimed. 

For the proof of (ii), similar steps with $\ket{01}-\ket{10}$ playing the role of $\ket{11}$ lead to:
\begin{align*}
	0  = &\sum_\ell c_\ell\, \tr_{\rms\setminus\rms_i}\bigg\{\left[(\bra{01}-\bra{10})_{\rms_i}\otimes\id_{\rms\setminus\rms_i}\right]\proj{\Psi_\ell}\times\\
	&\qquad\qquad\qquad \left[(\ket{01}-\ket{10})_{\rms_i}\otimes\id_{\rms\setminus\rms_i}\right]\bigg\}\\
	= &\sum_\ell c_\ell\, |\alpha^{(\ell)}_{0,1,i_3,\cdots,i_n}-\alpha^{(\ell)}_{1,0,i_3,\cdots,i_n}|^2
\end{align*}
This means that $\alpha^{(\ell)}_{0,1,i_3,\cdots,i_n}=\alpha^{(\ell)}_{1,0,i_3,\cdots,i_n}$ for all $\ell,i_3,i_4,\cdots,i_n$ and $\rms_i$. 
Hence, in the expansion of \cref{Eq.Expand}, if there is a term $\ket{i_1\cdots 01 \cdots i_n}$ where the ``$01$" appear at positions corresponding to an $\rms_i$, there must also be a term $\ket{i_1\cdots 10 \cdots i_n}$ with exactly the same amplitude. 
\end{proof}

\begin{theorem}\label{Thm:uniqueness}
For any tree graph with $n$ vertices that satisfies ~\cref{Eq.Wstate_marginal}, $\Omega_{n}(\gamma) = \gamma\proj{W_n} + (1-\gamma)\proj{0^n}$ is the unique global state and all the two-qubit reduced states are $\rho_n(\gamma)$.
\end{theorem}
\begin{proof}
For convenience, we define $\ket{m_1,m_2,\cdots,m_l}_n$ as the $n$-qubit state with a ``1" at positions $m_1,m_2,\cdots,m_l$ and $0$ elsewhere.
We first start with a linear chain, and suppose it has $n$ nodes and all of the $(n-1)$ edges are $\rho_n(\gamma)$. 
By ~\cref{lem.if_01_then_10}, we know that if any of the eigenstates $\ket{\Psi_\ell}$ has a contribution from $\ket{m_1,m_2,\cdots,m_l}_n$ (where $m_1<m_2<\cdots m_l$), there must also be an equal-amplitude contribution from both $\ket{m_1-1,m_2,\cdots,m_l}_n$ and $\ket{m_1+1,m_2,\cdots,m_l}_n$.  Repeating this argument iteratively eventually leads to the conclusion that there must also be a contribution from the term $\ket{m_2-1,m_2,\cdots,m_l}_n$ in $\ket{\Psi_\ell}$, which contradicts the  part (i) of \cref{lem.if_01_then_10}.
This means that each $\ket{\Psi_\ell}$ must lie in the span of $\ket{0^n}$ and $\{\ket{i}_k\}_{i=1,\cdots,n}$ and by part (ii) of \cref{lem.if_01_then_10}, all $\{\ket{i}_n\}_{i=1,\cdots,n}$ must  occur at the same time with the same amplitude, thereby giving
\begin{equation}\label{Eq.eigenstate}
	\ket{\Psi_\ell} = \beta^{(\ell)}_0\ket{0^n} + \beta^{(\ell)}_1\ket{W_n},\quad |\beta_0|^2+|\beta_1|^2=1.
\end{equation}

Again, imagine that $\rms_i$ being the first two qubits, then
\begin{align*}
	0 = &\,\bra{00}\rho_n(\gamma)\ket{\Psi^+} = \bra{00}\tr_{\rms\backslash\rms_i}(\rho)\ket{\Psi^+}\\
	= &\, \tr_{\rms\backslash\rms_i}\left[\left(\bra{00}_{\rms_i}\otimes\id_{\rms\setminus\rms_i}\right)\rho\left(\ket{\Psi^+}_{\rms_i}\otimes\id_{\rms\setminus\rms_i}\right)\right] \\
	= &\sum_\ell c_\ell\, \tr_{\rms\backslash\rms_i}\left[\left(\bra{00}_{\rms_i}\otimes\id_{\rms\setminus\rms_i}\ket{\Psi_\ell}\right)\left(\braket{\Psi_\ell}{\Psi^+}_{\rms_i}\otimes\id_{\rms\setminus\rms_i}\right)\right]\\
	= &\sum_\ell c_\ell\,\beta^{(\ell)}_0\sqrt{\tfrac{2}{n}}\beta_1^{(\ell)*}
\end{align*}
Hence, we have the constraint:
\begin{equation}\label{Eq.constraint}
	\sum_\ell c_\ell \beta^{(\ell)}_0\beta_1^{(\ell)*}=\sum_\ell c_\ell \beta^{(\ell)*}_0\beta_1^{(\ell)}=0
\end{equation}
Consequently, we see from \cref{Eq:Spectral_Decomposition_of_state} and \cref{Eq.eigenstate}, and \cref{Eq.constraint} that the global state is:
\begin{align}
	\rho = &\sum_\ell c_\ell \Big[ |\beta^{(\ell)}_0|^2\proj{0^n} + |\beta^{(\ell)}_1|^2\proj{W_n} \nonumber\\
	&\qquad + \beta^{(\ell)}_0\beta_1^{(\ell)*}\ket{0^n}\!\bra{W_n} + \beta^{(\ell)*}_1\beta_0^{(\ell)}\ket{W_n}\!\bra{0^n} \Big]\nonumber\\
	= & \sum_\ell c_\ell \left[ |\beta^{(\ell)}_0|^2\proj{0^n} + |\beta^{(\ell)}_1|^2\proj{W_n}\right],
\end{align}
which is a convex mixture of $\proj{0^n}$ and $\proj{W_n}$. Finally, using \cref{Eq.Wstate_marginal} and equating the two-qubit reduced states of $\rho$ with that required in \cref{eq.Rho_nGamma} immediately lead to:
\begin{equation*}
	 \sum_\ell c_\ell |\beta^{(\ell)}_0|^2 = 1-\gamma, \quad \sum_\ell c_\ell |\beta^{(\ell)}_1|^2 = \gamma,\quad \gamma\in [0,1].
\end{equation*}
Hence, the global state is necessarily
\begin{equation}
    \rho = \Omega_{n}(\gamma) = \gamma\proj{W_n} +(1-\gamma)\proj{0^n}.
\end{equation}
The above argument also holds for any $n$-node tree graph with all its $n-1$ edges set to $\rho_n(\gamma)$. To see this, it suffices to note that in a tree graph, there is always a unique path (chain) connecting any two nodes. We can then apply the above arguments for a chain to each of these paths to complete the analysis. As $\rho$ is clearly invariant under an arbitrary permutation of the $n$ subsystems, all its two-qubit reduced states are $\rho_n(\gamma)$. In particular, if $\rmt\not\in\calS$ is a two-qubit marginal, we must also have $\rhot=\rho_n(\gamma)$.
\end{proof}
Note that our \cref{Thm:uniqueness} generalizes the uniqueness result of \cite{Parashar.PRA.2009,wu2014} where the global state is the $n$-qubit $W$-state $\ket{W_n}$.

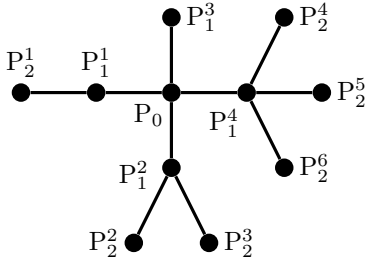
\begin{figure}[ht]
    \centering
\begin{tikzpicture}
\begin{scope}[every node/.style={circle,fill,inner sep=0pt, minimum size = 1pt, scale = 0.7}]
    \node (A) at (0,3) {O};
    \node (B) at (1,3) {O};
    \node (C) at (2,3) {O};
    \node (D) at (3,3) {O};
    \node (E) at (4,3) {O};
    \node (F) at (2,4) {O};
    \node (G) at (2,2) {O};
    \node (H) at (2.5,1) {O};
    \node (I) at (1.5,1) {O};
    \node (J) at (3.5,4) {O};
    \node (K) at (3.5,2) {O};
\end{scope}

\begin{scope}
    \node at (1.7, 2.7) {$\mathrm{P}_0$};
    \node at (1, 3.4) {$\mathrm{P}_1^1$};
    \node at (1.5, 1.9) {$\mathrm{P}_1^2$};
    \node at (2.4, 4) {$\mathrm{P}_1^3$};
    \node at (2.7, 2.6) {$\mathrm{P}_1^4$};
    \node at (0, 3.4) {$\mathrm{P}_2^1$};
    \node at (1.1, 1) {$\mathrm{P}_2^2$};
    \node at (2.9, 1) {$\mathrm{P}_2^3$};
    \node at (3.9, 4) {$\mathrm{P}_2^4$};
    \node at (4.4, 3) {$\mathrm{P}_2^5$};
    \node at (3.9, 2) {$\mathrm{P}_2^6$};
\end{scope}

\begin{scope}[every node/.style={fill=white,circle},
              every edge/.style={draw=black,very thick}]
    \path [-] (A) edge (B);
    \path [-] (B) edge (C);
    \path [-] (C) edge (D);
    \path [-] (D) edge (E);
    \path [-] (C) edge (F);
    \path [-] (C) edge (G);
    \path [-] (G) edge (H);
    \path [-] (G) edge (I);
    \path [-] (D) edge (J);
    \path [-] (D) edge (K);
\end{scope}
\end{tikzpicture}
    \caption{To see how the proof above applies to any tree graph, suppose we start from node ${\rm P}_0$ in this example. Then we build up the possible eigenstate by applying Lemma~\ref{lem.if_01_then_10} to all nodes ${\rm P}_{k}^{j}$ that are distance $k$ away from ${\rm P}_{0}$. Because there is a unique path between ${\rm P}_0$ and any other node in the tree graph, this leads to the same conclusion as a linear chain.}
    \label{fig.proofUniqueMixedWTree}
\end{figure}

\section{Finding the (meta)transitivity region of overlapping Werner states}\label{app.RegionWerner}

Consider a qudit tripartite system $\rabc$ for $d \ge 3$.
Ref.~\cite{johnsonviola2013} describes the conditions for three Werner states in $\rab$,  $\rac$, and $\rbc$ to be compatible. In Ref.~\cite{johnsonviola2013}, they parameterize the Werner state according to
\begin{equation}
    W_{d}(\psi^-) = \tfrac{d}{d^2-1} \left[ (d-\psi^-)\tfrac{1}{d^2}\mathbb{I} + (\psi^- - \tfrac{1}{d}) \tfrac{1}{d}V \right]
\end{equation}
where $V$ is the swap operator $V\ket{\alpha}\ket{\beta} = \ket{\beta}\ket{\alpha}$ and 
\begin{equation}\label{eq:WernerPsiMinus}
	\psi^- = \tr[VW_{d}(\psi^-)].
\end{equation}

Ref.~\cite{johnsonviola2013} showed that three qudit Werner states $\psi^{-}_{\rab}, \psi^{-}_{\rac},\psi^{-}_{\rbc}$ are compatible if and only if
the point $(\psi^{-}_{\rab},\psi^{-}_{\rbc},\psi^{-}_{\rac})$
lies within the bicone described by
\begin{equation}
\label{eq.WernerBicone}
    1 \pm \psi^{-}_\mathrm{ave} \ge \tfrac{2}{3} \lvert \psi^{-}_{\rbc} + \omega\psi^{-}_{\rac} + \omega^2 \psi^{-}_{\rab}\rvert,
\end{equation}
where $\omega = \exp(\tfrac{2\pi i}{3})$ and
\begin{equation}
    \psi^{-}_\mathrm{ave} = \tfrac{1}{3}(\psi^{-}_{\rab} +\psi^{-}_{\rac}+\psi^{-}_{\rbc}).
\end{equation}

In terms of the parameter $v$ in 
Eq.~(3), we have $\psi^{-} = 2v-1$, so the compatibility conditions become
\begin{align}
\label{eq.compatibleWerner}
    \nonumber
    \tfrac{2}{3}(v_{\rab} + v_{\rac} + v_{\rbc}) \ge \mathcal{F} \\
2 - \tfrac{2}{3}(v_{\rab} + v_{\rac} + v_{\rbc}) \ge \mathcal{F}
\end{align}
where
\begin{equation}
    \mathcal{F} :=
\tfrac{2}{3}\sqrt{3(v_{\rac} - v_{\rab})^2 + (2v_{\rbc} - v_{\rab} - v_{\rac})^2}
\end{equation}

To find the metatransitivity region, we need to find the range of compatible $v_{\rbc}$ when given $v_{\rab}$ and $v_{\rac}$ and solve for when the boundary $v_{\rbc} =\tfrac{1}{2}$.

For the first inequality in Eq.~(\ref{eq.compatibleWerner}), if we square both sides and simplify, we obtain
\begin{gather}
    v_{\rbc}^2 - 2v_{\rbc}(v_{\rab} + v_{\rac}) + (v_{\rab}-v_{\rab})^2 \le 0.
\end{gather}
Next we complete the square for $v_{\rbc}$ to get
\begin{equation}\label{Eq.Werner.Parabola1}
    [v_{\rbc} - (v_{\rab} + v_{\rac})]^2 \le 4v_{\rab} v_{\rac}.
\end{equation}
The desired boundary is given by taking the equality and substituting $v_{\rbc} = \tfrac{1}{2}$.

Similarly, for the second inequality in Eq.~(\ref{eq.compatibleWerner}), if we square both sides and simplify, we obtain
\begin{gather}
    \nonumber
    v_{\rbc}^2 + 2v_{\rbc} - 2v_{\rbc}(v_{\rab}+v_{\rac})
    + (v_{\rab}-v_{\rac})^2 \\
    \le 3 - 2(v_{\rab}+v_{\rac}).
\end{gather}
This time we complete the square for $(v_{\rbc}+1)$ to get
\begin{gather}\label{Eq.Werner.Parabola2}
    \left[ (v_{\rbc} + 1) - (v_{\rab} + v_{\rac}) \right]^2
    \le 4(1-v_{\rab})(1-v_{\rac}).
\end{gather}
The desired boundary is given by taking the equality and substituting $v_{\rbc} = \tfrac{1}{2}$.

\begin{figure}[tbp]
    \centering
    \includegraphics[width=0.65\linewidth]{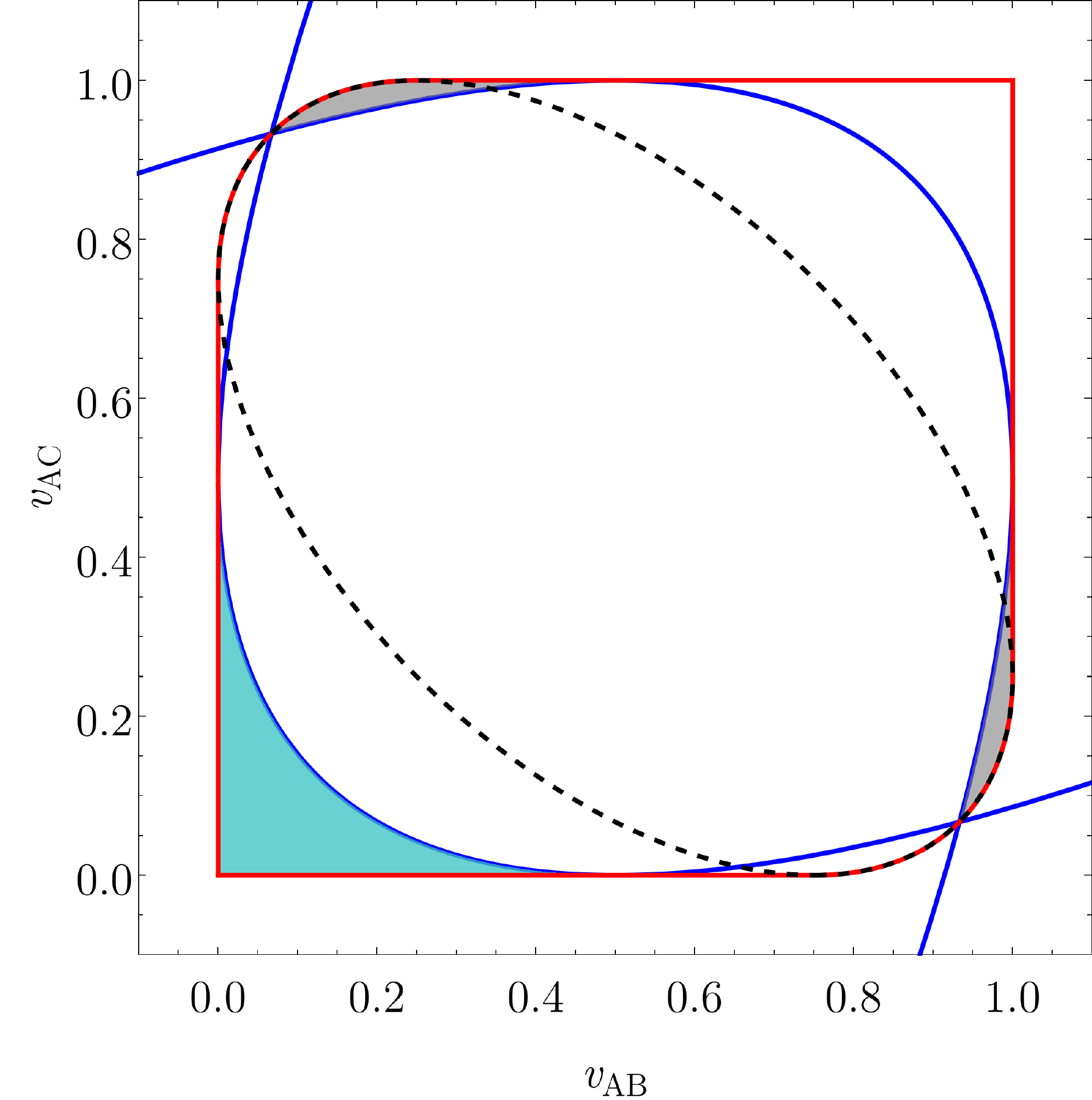}
    \caption{Parameter space for a pair of Werner state marginals with, respectively, weight $v_{\rab}$ and $v_{\rac}$ on the symmetric subspace. 
    The plot here shows the parabolas of \cref{Eq.Werner.Parabola1}  and \cref{Eq.Werner.Parabola2} as well as the ellipse of \cref{eq.WernerEllipse}. 
    Fig.~2 is a simplified version of the current plot.
    }
    \label{fig:MetatransitivityQuditWerner-detail}
\end{figure}

Ref.~\cite{johnsonviola2013} specifies the compatible region for a pair of Werner states obtained from projecting the bicone onto a plane. This compatible region is given by $\psi^{-}_{\rab},\psi^{-}_{\rac} \ge -\tfrac{1}{2}$
or $\psi^{-}_{\rab},\psi^{-}_{\rac} \le \tfrac{1}{2}$, or
the pair satisifies
\begin{equation}
\label{eq.WernerPsiMinusEllipse}
    (\psi^{-}_{\rab} + \psi^{-}_{\rac})^2 +\tfrac{1}{3}(\psi^{-}_{\rab} - \psi^{-}_{\rac})^2 \le 1.
\end{equation}
In our parameters,  this translates to the convex hull of the points $(0,0),(1,1)$ and all the points contained in the ellipse 
\begin{equation}
\label{eq.WernerEllipse}
    (v_{\rab} + v_{\rac}-1)^2 + \tfrac{1}{3}(v_{\rab} - v_{\rac})^2 = \tfrac{1}{4}.
\end{equation}
Finally we find that the parabolas will divide the compatible region into seven areas. It is enough to check if a point inside each area to determine if the area exhibits metatransitivity.

For $d=2$, only the cone given by the minus sign in Eq.~(\ref{eq.WernerBicone}) is compatible. This leads to a compatible region for $(\psi^{-}_{\rab},\psi^{-}_{\rac})$ that is given by
$\psi^{-}_{\rab},\psi^{-}_{\rac} \ge -\tfrac{1}{2}$ or Eq.~(\ref{eq.WernerPsiMinusEllipse}). This translates to the convex hull of $(1,1)$ and the ellipse of Eq.~(\ref{eq.WernerEllipse}). To understand why this happens, observe that the projection onto the qubit antisymmetric subspace corresponds to the maximally entangled singlet state $\tfrac{1}{\sqrt{2}}(\ket{01}-\ket{10})$, so for small values of $v_{\rab}$ and $v_{\rac}$, monogamy of entanglement prohibits them from being compatible.

\section{Finding the metatransitivity region of overlapping isotropic states}\label{app.RegionIsotropic}

Consider a qudit tripartite system $\rabc$ for $d \ge 3$.
Ref.~\cite{johnsonviola2013} describes the conditions for two isotropic states in $\rab$ and  $\rac$, and $\rbc$ to be compatible. In Ref.~\cite{johnsonviola2013}, they parameterize the isotropic state according to
\begin{equation}
    \mathcal{I}_{d}(\phi^{+}) = \tfrac{d}{d^2-1} \left[ (d-\phi^+)\tfrac{1}{d^2}\mathbb{I} + (\phi^+ - \tfrac{1}{d}) \proj{\Phi_d} \right]
\end{equation}
where $\ket{\Phi_d} = \tfrac{1}{\sqrt{d}}\sum_{i}\ket{i,i}$ and
$d\phi^{+} = \bra{\Phi_d}\mathcal{I}_d(p)\ket{\Phi_d}$ is, up to a constant of $d$, the fully entangled fraction of $\mathcal{I}_{d}(\phi^{+})$. Meanwhile the Werner state in $\rbc$ is written in terms of $\psi^-$ in Eq.~(\ref{eq:WernerPsiMinus}).
Ref.~\cite{johnsonviola2013} showed that for $d \ge 3$ the $\phi^{+}_{\rab}, \phi^{+}_{\rac}$ and $\psi^{-}_{\rbc}$ are compatible if the point $(\phi^{+}_{\rab},\phi^{+}_{\rac},\psi^{-}_{\rbc})$ lies within the convex hull of $(0,0,-1)$ and the cone given by
\begin{align}
\nonumber
    \phi^{+}_{\rab} &+ \phi^{+}_{\rac} - \psi^{-}_{\rbc} \le d, \\
    \nonumber
    1 &+ \phi^{+}_{\rab} + \phi^{+}_{\rac} - \psi^{-}_{\rbc}  \\
    \nonumber
    &\ge \left\lvert d(\psi^{-}_{\rbc} - 1) +\sqrt{\tfrac{2d}{d-1}}(e^{i\theta}\phi^{+}_{\rab} + e^{-i\theta}\phi^{+}_{\rac}) \right\rvert, \\
    e^{\pm i\theta} &= \pm i\sqrt{\tfrac{d+1}{2d}} + \sqrt{\tfrac{d-1}{2d}}. 
\end{align}

In terms of the fully entangled fraction $p = \tfrac{1}{d}\phi^{+}$ for the isotropic states and $v = \tfrac{1}{2}(\psi^{-}+1)$ for the Werner states, the compatibility conditions become
\begin{subequations}
\begin{align}
\label{Eq.LinConstr}
    p_{\rab} &+ p_{\rac} - \tfrac{1}{d}(2v_{\rbc} - 1) \le 1, \\
    \nonumber
    2 &+ d(p_{\rab} + p_{\rac}) - 2 v_{\rbc} 
    \ge \sqrt{\mathcal{R}_1 + \mathcal{R}_2}, \\
    \nonumber
 \mathcal{R}_1 &= [d (p_{\rab}+p_{\rac})+d (2 v_{\rbc}-2)]^2, \\
 \mathcal{R}_2 &= \tfrac{d^2 (d+1)}{d-1} \left(p_{\rab}-p_{\rac}\right)^2.
 \label{Eq.NonlinearConstr}
\end{align}
\end{subequations}

Similar to what we did for the Werner states, we want to solve for the condition on $p_{\rab}$ and $p_{\rac}$ such that $v_{\rbc} = \frac{1}{2}$ is on the boundary of the compatible Werner states.
Let $\mathcal{V} = 2v_{\rbc}-2$ and 
$\mathcal{P} = d(p_{\rab} + p_{\rac})$.
Taking \cref{Eq.NonlinearConstr} and squaring both sides, we obtain
\begin{equation}
    (\mathcal{V} - \mathcal{P})^2 \ge (d\mathcal{V} + \mathcal{P})^2
    + \tfrac{d+1}{d-1}\mathcal{P}^2 - \left(\tfrac{d+1}{d-1} \right)4d^2 p_{\rab}p_{\rac}.
\end{equation}
After some algebra this can be simplified into
\begin{equation}
    [\mathcal{P} + (d-1)\mathcal{V}]^2 \le 4d^2 p_{\rab}p_{\rac}.
\end{equation}
The desired boundary is obtained by taking the equality and setting $v_{\rbc} = \tfrac{1}{2}$, which implies $\mathcal{V} = -1$ and leads to the parabola
\begin{equation}
\label{eq.IsotropicMetatransitivityParabola}
p_{\rab} p_{\rac}=\left[\frac{d(p_{\rab}+p_{\rac})-(d-1)}{2 d}\right]^2.
\end{equation}

\begin{figure}[t]
    \centering
    \includegraphics[width=0.65\linewidth]{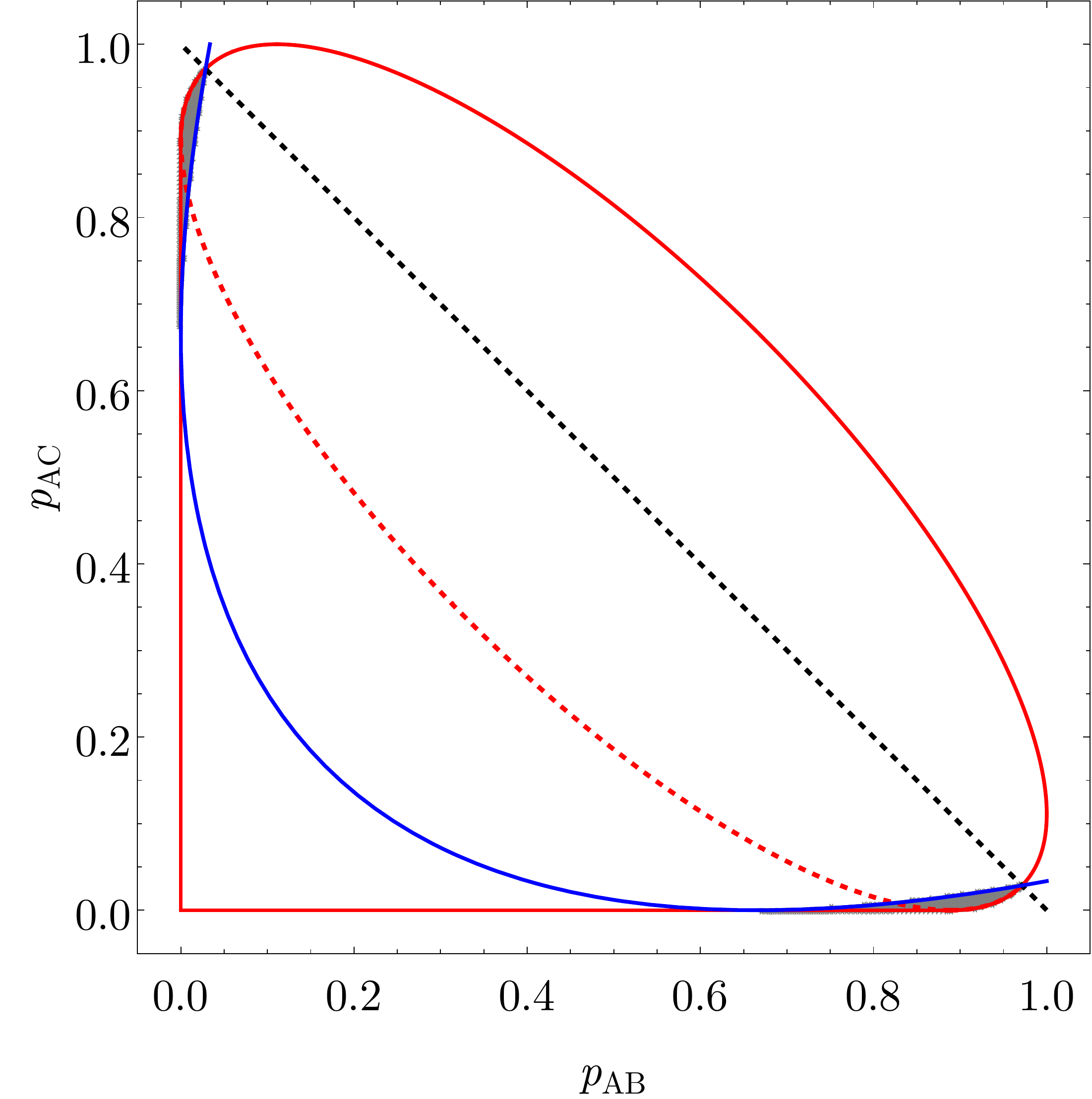}
    \caption{Parameter space for a pair of isotropic state marginals with, respectively, fully entangled fraction $p_{\rab}$ and $p_{\rac}$. 
    The compatible region for the pair is enclosed by the solid red line. The plot here shows the parabolas of \cref{eq.IsotropicMetatransitivityParabola}  and the ellipse of \cref{Eq.isotropic.ellipse}. 
    Fig.~3 is a simplified version of the current plot.} 
    \label{fig:MetatransitivityQuditIsotropic-detail}
\end{figure}

Ref.~\cite{johnsonviola2013} specifies the compatible region for a pair of isotropic states to be the region given by the convex hull of $(\phi^+_{\rab}, \phi^{+}_{\rac}) = (0,0)$ and the ellipse
\begin{equation}\label{Eq.isotropic.ellipse}
   \frac{(\tfrac{1}{d}\phi^{+}_{\rab} + \tfrac{1}{d}\phi^{+}_{\rac} - 1)^2}{\tfrac{1}{d^2}} + \frac{(\tfrac{1}{d}\phi^{+}_{\rab} - \tfrac{1}{d}\phi^{+}_{\rac})^{2}}{\tfrac{d^2-1}{d^2}} = 1,
\end{equation}
which in our parameters becomes the convex hull of
the point $(p_{\rab},p_{\rac}) = (0,0)$ and the ellipse
\begin{equation}
(p_{\rab}+p_{\rac}-1)^2 + \tfrac{1}{d^2-1}(p_{\rab}-p_{\rac})^2 = \tfrac{1}{d^2}.
\end{equation}
Finally, we verify that the parabola in Eq.~(\ref{eq.IsotropicMetatransitivityParabola}) divides the compatible region into four areas, and that the metatransitivity region obtained with this parabola matches the one that is obtained numerically for $d \le 5$ up to numerical precision.

\section{Other explicit examples}
\label{App.Examples}

For ease of reference, we summarize in \cref{Tbl.Summary.Eg} the nature of the various explicit examples presented in this Appendix.

\begin{table}[h!]
\begin{ruledtabular}
\begin{tabular}{c||c|cc|c|cc}
Example & $n$  & $d_{\rms_i}$ & $\sigma_{\rms_i}^\Gamma\succeq0$? & $d_{\rmt}$ &  $|\displaystyle 1-\min_{\rhos} \bra{\psi}\rhos\ket{\psi}|$ \\ \hline
\ref{App.3Qubit.SelfComp} & 3  & 2 & None & $2\times 2$ & 0\\
\ref{App.Genuine4} & 4  & 2 & None & $\{2\times 2\}^3$ & $\approx 10^{-11}$\\
\ref{App.k-ext} & 4 to 7  & 2 & All & $2\times 2$ & - \\
\ref{App.BE} & 3 & 3 & All & $3\times 3$ & - \\
\ref{App.GME} & 4  & 2 & All & $2\times 2\times 2$ & $\approx 10^{-9}$\\
\end{tabular}
\end{ruledtabular}
\centering \caption{\label{Tbl.Summary.Eg} Table summarizing various aspects of the explicit examples of (meta)transitivity presented in \cref{App.Examples}. From left to right, we list the subsection in which the example is presented, the local Hilbert space dimension $d_{\rms_i}$ of the input marginal $\sigma_{\rms_i}$, whether these input marginals are PPT, the dimensions $d_{\rmt}$ of the target system Hilbert space (e.g., $2\times2$ means $\rmt$ is a two-qubit system, whereas $\{2\times 2\}^3$ means three different two-qubit target systems have been considered), and if applicable, the minimum compatible fidelity of the joint state $\rhos$ with respect to a known compatible pure state $\ket{\psi}$.}
\end{table}



\subsection{Three-qubit transitivity from symmetric extensions}
\label{App.3Qubit.SelfComp}

Apart from the Werner state and the isotropic state marginals, here, we show that ETP can also be solved for a four-parameter family of two-qubit marginals. To this end, consider the two-qubit state
\begin{align}
\label{eq.selfcompChoi}
\sigma_{\rab} =
\begin{pmatrix}
a_0^2                   & 0                                 &  0                              & a_0 \frac{b}{\sqrt{2}} \\
0                       & a_1^2                             & e^{it}  a_1 \frac{b}{\sqrt{2}}  & 0 \\
0                       & e^{-it} a_1  \frac{b}{\sqrt{2}}   &  \frac{b^2}{2}                  & 0 \\
a_0 \frac{b}{\sqrt{2}}  &0                                  & 0                               & \frac{b^2}{2}
\end{pmatrix},
\end{align}
where $a_0,a_1,b \in [-1,1]$ and $t\in [0,2\pi]$.
It can be shown that Eq.~(\ref{eq.selfcompChoi}) is, up to normalization, the Choi representation of a single-qubit selfcomplementary quantum operation~\cite{Smaczynski2016}.

Computing the eigenvalues of the partial transpose of Eq.~(\ref{eq.selfcompChoi}), the smallest eigenvalue is given by
$\lambda_\mathrm{min} = \tfrac{1}{4} \left( 2a_i^2 + b^2 - \sqrt{4a_i^4 + 8a_j^2 b^2 - 4 a_i^2 b^2 + b^4} \right)$
for $i=0,j=1$ and vice-versa. Thus, Eq.~(\ref{eq.selfcompChoi}) is entangled when $|b|\in (0,1)$ and $|a_0| \ne |a_1|$ for $|a_0|,|a_1|\in [0, 1)$.
Next, we prove that entangled $\sigma_{\rab}$ has the pure, unique symmetric extension
\begin{align}
\label{eq.uniqSol}
\ket{\Psi}_{\rabc} = a_0 \ket{000} + a_{1}e^{it} \ket{011} + b\ket{1}\ket{\Psi^+},
\end{align}
where $\ket{\Psi^+}= \tfrac{1}{\sqrt{2}} (\ket{01} + \ket{10})$. It is easy to check that $\ket{\Psi}$ has the correct marginals, so it remains to show that it is unique.
For this, we show that the eigenstates of an arbitrary qubit tripartite state $\rho_{\rabc}$ must have a particularly structure in order to produce the correct marginal states $\rho_{\rab} = \sigma_{\rab} = \rho_{\rac}$. The proof may be of independent interest but so as to not detract attention from the discussion here, we postpone the details to Supplementary Note~\ref{app.uniqueSymExtSelfComp}.

Finally, because Eq.~(\ref{eq.uniqSol}) is the unique joint state, we obtain transitivity by solving for the case when its BC marginal is non-PPT. It is straightforward to verify that the characteristic polynomial of $\rho_{\rbc}^{\Gamma}$ can be factorized into
$\left( \tfrac{b^2}{2} \pm a_0 a_1 - x\right)$  and $\left[a_0^2 a_1^2 - \tfrac{b^2}{2} - \left(a_0^2 + a_1^2\right)x + x^2\right]$,
which yields a negative root when $|b| \ne \sqrt{2|a_0 a_1|}$.

\subsection{Genuine four-qubit transitivity with entangled marginals}
\label{App.Genuine4}

For completeness, we provide here a four-qubit state $\ket{\chi}_{\rm ABCD}$ with entangled marginals for AB, BC, and CD such that they exhibit the same kind of genuine four-party effect displayed by the example with all separable marginals (case (d) in Fig.~4) given by Eq.~(6) of the main text:
\begin{align}
\nonumber
\ket{\chi}_{\rm ABCD} &= \tfrac{1}{\sqrt{N}}\left(
    -\tfrac{1}{5},-\tfrac{1}{12},-\tfrac{1}{93},-\tfrac{2}{9},\tfrac{3}{10},\tfrac{1}{6},-\tfrac{1}{4},-\tfrac{2}{3}, \right. \\
    &\qquad
    \left. \tfrac{1}{11},-\tfrac{3}{11},\tfrac{1}{7},-\tfrac{1}{6},-\tfrac{1}{6},-\tfrac{1}{4},\tfrac{2}{9},\tfrac{1}{7} \right)\Tp,
\end{align}
where $N$ is a normalization constant. Imposing the AB, BC, and CD marginals of $\ket{\chi}_{\rm ABCD}$ in Eq.~(\ref{eq.transitivitySDP}) leads, respectively, to
$\lambda_{\rm AD}^\star \approx -0.0788, \lambda_{\rac}^\star \approx -0.1344$, and 
$\lambda_{\rm BD}^\star \approx -0.0553$. These can also be verified by noting that $\ket{\chi}_{\rm ABCD}$ appears to be the unique state compatible with these marginals.

\subsection{\texorpdfstring{$k$}{k}-qubit metatransitivity with separable marginals for \texorpdfstring{$k$}{k} from 4 to 7}
\label{App.k-ext}

Next, we present some examples that may be extended to a more complicated setting. We begin with a four-qubit metatransitivity example where the separable marginals AB, BC, and CD can be used to infer the entanglement in AD. Let $\mathcal{B}(\vec{w})$ be a  Bell-diagonal two-qubit state where $\vec{w}$ is the vector of convex weights of the Bell states $\{\ket{\Phi^\pm} = \tfrac{1}{\sqrt{2}}(\ket{00}\pm\ket{11}),\ket{\Psi^\pm}=\tfrac{1}{\sqrt{2}}(\ket{01}\pm\ket{10})\}$, in that order. Take the marginal states $ \mathcal{B}(\vec{w}_{\rab}), \mathcal{B}(\vec{w}_{\rbc})$, and $\mathcal{B}(\vec{w}_{\rm CD})$, where

\begin{align}
\label{eq.SepABSepBCSepCDEntADBellDiag}
\nonumber
\vec{w}_{\rab} &= \tfrac{1}{10^4}\left( 1363, 4552, 610, 3475 \right), \\
\nonumber
\vec{w}_{\rbc} &= \tfrac{1}{10^4}\left( 1819, 4153, 3957, 71 \right), \\
\vec{w}_{\rm CD} &= \tfrac{1}{10^4}\left( 4440, 3209, 2028, 323 \right).
\end{align}
These marginals are separable because a Bell-diagonal state is separable {\em iff} all $w_i$ are less than $\tfrac{1}{2}$~\cite{Horodecki.PRA.1996}. Using Eq.~(\ref{eq.transitivitySDP}), we obtain 
$\lambda_{\rm AD}^{\star} \approx -0.0020$.
Since the optimal joint state has separable Bell-diagonal states in AC and BD, the metatransitivity of entanglement is not possible in those marginals.

Remarkably, the same Bell-diagonal states can be used to exhibit 5-qubit metatransitivity by taking the marginals $\mathcal{B}(\vec{z})$:
\begin{align}
    \vec{z}_{\rab} &= \vec{w}_{\rab}, &
    \vec{z}_{\rbc} &= \vec{z}_{\rm CD} = \vec{w}_{\rbc}, &
    \vec{z}_{\rm DE} &= \vec{w}_{\rm CD}. 
\end{align}
Indeed, we obtain $\lambda^\star_{\rm AE} \approx -0.1165$, thus exhibiting metatransitivity between the ends of the chain from A to E in~\cref{fig.5to6and7qubitsMetatransitivity}.

We next present an example of five-qubit metatransitivity that may be extended in a different manner.  The four input Bell-diagonal marginals $\mathcal{B}(\vec{q})$ are
\begin{align}
\label{eq.SepABSepBCSepCDSepDEEntAEBellDiag}
\nonumber
\vec{q}_{\rab} &= \tfrac{1}{10^4}\left( 566, 4203, 3933, 1298 \right), \\
\nonumber
\vec{q}_{\rbc} &= \tfrac{1}{10^4}\left( 3252, 4614, 2068, 66 \right), \\
\nonumber
\vec{q}_{\rm CD} &= \tfrac{1}{10^4}\left( 4324, 3437, 323, 1916 \right), \\
\vec{q}_{\rm DE} &= \tfrac{1}{10^4}\left( 818, 4430, 503, 4249 \right).
\end{align}
From Eq.~(\ref{eq.transitivitySDP}) we obtain $\lambda^\star_{\rm AE} \approx -0.0379$. 
Interestingly, we can use these Bell-diagonal states to get metatransitivity examples for six and seven qubits from a tree graph (see \cref{fig.5to6and7qubitsMetatransitivity}) of separable marginals. For the six-qubit example, we keep the marginals of \cref{eq.SepABSepBCSepCDSepDEEntAEBellDiag} and add another node F with $\vec{q}_{\rm BF} = \vec{q}_{\rbc}$, which again gives $\lambda^\star_{\rm AE} \approx -0.0379$. In the seven-qubit case, we keep all these marginals and add a node G with $\vec{q}_{\rm BG} = \vec{q}_{\rbc}$, this time around giving  $\lambda^\star_{\rm AE} \approx -0.0402$. We also note that \cref{eq.SepABSepBCSepCDSepDEEntAEBellDiag}  does not show metatransitivity in the other bipartite marginals, as can be seen from the separable marginals in the optimal global state for the metatransitivity in ${\rm AE}$.

\begin{figure}[t]
    \centering
\begin{tikzpicture}
\begin{scope}[every node/.style={circle,fill,inner sep=0pt, minimum size = 1pt, scale = 0.7}]
    \node (A) at (0,0.3) {A};
    \node (B) at (0.75,0.3) {B};
    \node (C) at (1.5,0.3) {C};
    \node (D) at (2.25,0.3) {D};
    \node (E) at (3.0,0.3) {E};
    \node (F) at (0.2, -0.45) {F};
     \node (G) at (1.3, -0.45) {G};
\end{scope}
\begin{scope}
    \node at (0,0.7)  {$\ra$};
    \node at (0.75,0.7)  {$\rb$};
    \node at (1.5,0.7)  {$\rc$};
    \node at (2.25,0.7)  {$\rd$};
    \node at (3.0,0.7)  {$\rm E$};
    \node at (-0.2,-0.45)  {$\rm F$};
     \node at (1.7,-0.45)  {$\rm G$};
\end{scope}
\begin{scope}[every node/.style={fill=white,circle},
              every edge/.style={draw=black,very thick}]
    \path [-] (A) edge (B);
    \path [-] (B) edge (C);
    \path [-] (C) edge (D);
    \path [-] (D) edge (E);
    \path [-] (B) edge (F);
    \path [-] (B) edge (G);
\end{scope}
\end{tikzpicture}
    \caption{A tree graph showing how the metatransitivity example of \cref{eq.SepABSepBCSepCDSepDEEntAEBellDiag} for qubits A, B, C, D, and E may be extended to six and seven qubits. The entanglement in AE can be certified by specifying  the chain of neighboring two-qubit marginals from A to E alone, or together with BF and/or BG.}
    \label{fig.5to6and7qubitsMetatransitivity}
\end{figure}
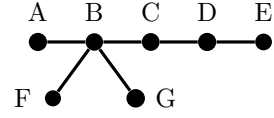

\subsection{Three-qutrit transitivity from bound entangled states}
\label{App.BE}

Here we provide two examples of transitivity involving marginal states that are PPT bound entangled~\cite{Horodecki:PRL:1998}. For this, we consider the bound entangled state obtained from the unextendible product basis (UPB) known as $\mathrm{Tiles}$~\cite{BennettDiVincenzo1999}:
\begin{gather}
\nonumber
    \ket{T_0} = \tfrac{1}{\sqrt{2}}\ket{0}(\ket{0}-\ket{1}), \qquad
    \ket{T_2} = \tfrac{1}{\sqrt{2}}\ket{2}(\ket{1}-\ket{2}), \\
    \nonumber
    \ket{T_1} = \tfrac{1}{\sqrt{2}}(\ket{0}-\ket{1})\ket{2}, \qquad
    \ket{T_3} = \tfrac{1}{\sqrt{2}}(\ket{1}-\ket{2})\ket{0}, \\
    \ket{T_4} = \tfrac{1}{\sqrt{2}}(\ket{0} + \ket{1} + \ket{2})(\ket{0} + \ket{1} + \ket{2}).
\end{gather}
The bound entangled state $\rho_\mathrm{Tiles}$ is obtained by taking the normalized projector onto the subspace complementary to the $\mathrm{Tiles}$ UPB:
$\rho_{\mathrm{Tiles}} = \tfrac{1}{4}\left(\id - \sum_{i=0}^{4} \proj{T_i}\right).$
Now if we employ \cref{eq.transitivitySDP} with marginals $\sigma_{\rab} = \sigma_{\rbc} = \rho_{\mathrm{Tiles}} $, we find the optimal value $\lambda_{\rac}^\star \approx -0.1194$, thus certifying the transitivity in $\rac$ given marginal states in $\rab$ and $\rbc$ that are bound entangled.

For the second example, we consider the UPB known as $\mathrm{Pyramid}$, which is given by
\begin{equation}
    \ket{P_j} = \ket{p_j}\otimes\ket{p_{2j\text{ mod }5}}, \qquad j=0,\ldots, 4,
\end{equation}
where $\ket{p_j}$ are states that form the base of a regular pentagonal pyramid in  $\mathbb{R}^3$:
\begin{equation}
    \ket{p_{j}} = \tfrac{2}{\sqrt{5 + \sqrt{5}}} \left( \cos\tfrac{2\pi j}{5}, \sin\tfrac{2\pi j}{5}, \tfrac{\sqrt{1+\sqrt{5}}}{2} \right)\Tp,
    \,\, j = 0,\ldots,4.
\end{equation}
The corresponding bound entangled state is
$\rho_{\mathrm{Pyramid}} = \tfrac{1}{4}\left(\id - \sum_{j=0}^{4} \proj{P_j}\right).$
Solving \cref{eq.transitivitySDP} with marginals $\sigma_{\rab} = \sigma_{\rbc} = \rho_{\mathrm{Pyramid}}$, we obtain $\lambda_{\rac}^\star \approx -0.1094$.

Interestingly, we observe a similar type of transitivity with the marginals $\sigma_{\rab} = \sigma_{\rbc}$ set to any of the $10^{5}$ randomly generated bound entangled states from the six-parameter family of all two-qutrit UPBs~\cite{DiVincenzoMor2003}.

\subsection{Four-qubit transitivity for genuine tripartite entanglement}
\label{App.GME}

Here, we give an example where a collection of two-qubit marginals imply the presence of genuine tripartite entanglement in a three-qubit marginal.
To this end, consider the AB, AC, and AD marginals arising from the four-qubit state
 \begin{align}
 \nonumber
  \ket{\Psi}_{\rm ABCD} &=
  \bigg( 
     \tfrac{1}{12},\tfrac{1}{9},0,\tfrac{1}{6},\tfrac{1}{9},\tfrac{1}{9},0,0,0,\tfrac{\sqrt{42}}{9}, \\
     &\qquad
     -\tfrac{1}{3},-\tfrac{1}{12},-\tfrac{1}{4},-\tfrac{1}{12},-\tfrac{1}{3},\tfrac{1}{3} \bigg)\Tp.
 \end{align}
 For these marginals, the smallest compatible values of $M$ and $N$ defined in \cref{Eq.defMandN}, respectively,
are
 \begin{align}
     M(\rho_{\rm BCD}) &\approx 1.8606 \text{ and } N(\rho_{\rm BCD}) \approx  1.8008. 
 \end{align}
In this case, however, the biseparable upper bound for the criterion of~\cite{Li2017}, see \cref{App.GME.Criterion}, is  $\beta = \tfrac{1+2d}{3}=\tfrac{5}{3}\approx1.667$, which is clearly violated. Thus, the BCD marginal given the aforementioned marginals of AB, AC, and AD must be genuinely tripartite entangled. Note that the global state compatible with these marginals again appear to be unique, see \cref{Tbl.Summary.Eg}.

\section{Extending metatransitivity examples to more parties}\label{app.MetaMoreParties}

Here we show how to extend an example of metatransitivity for $n$-parties to one involving $n+k$ parties, for arbitrary $k$.
Suppose we have an $n$-partite system $\rms$ with marginal states $\mathcal{S} = \{ \sigma_{\rms_i} \}$ and let $\rmt$ be some target marginal system in $\rms$ such that for some entanglement witness $\W$ we have that $\rhot = \tr_{\rms \backslash \rmt}(\rhos)$ and
$\W(\rhot) < 0$ for all joint states $\rhos$ compatible with $\mathcal{S}$.
Let $\rhos^\star$ denote the joint state with
\begin{equation}
    \lambda := \max_{\rhos} \W[\tr_{\rms \backslash \rmt}(\rhos)] = 
    \W[\tr_{\rms \backslash \rmt}(\rhos^\star)]
\end{equation}
from \cref{eq.transitivitySDP}. We assume metatransitivity in $\rmt$, so $\lambda < 0$.

Let ${\rm R}$ be the $(n+k)$-partite system such that ${\rm R} \backslash \rms = {\rm K}$, that is, ${\rm K}$ is the $k$-partite marginal system of ${\rm R}$ that is disjoint from the $n$-partite $\rms$. 
Let $\mathcal{R} = \mathcal{S} \cup \mathcal{K}$ where
$\mathcal{K} = \{ \tau_{{\rm R}_i} \}$ and ${\rm R}_i$ are marginal systems of ${\rm R}$ that are distinct (but not necessarily disjoint) from the marginal systems involved in $\mathcal{S}$.
To avoid trivial situations, we assume the marginals specified in $\mathcal{K}$ are compatible with those already given in $\mathcal{S}$.

Consider the following metatransitivity problems for ${\rm R}$:
\begin{gather}
\nonumber
\mu_1 :=
\max_{\rho_{\rm R}} \W\left[\tr_{{\rm R} \backslash \rmt}(\rho_{{\rm R}})\right], 
\text{ s.t. }  \tr_{{\rm R}\backslash \rms_{i}}(\rho_{{\rm R}}) = \sigma_{\rms_{i}} \forall\,\, \rms_{i}\in\calS, \\
\tr_{{\rm R} \backslash {\rm R}_i}(\rho_{{\rm R}}) = \tau_{{\rm R}_i} \forall\,\, {\rm R}_{i}\in \mathcal{K}, \,\, \rho_{{\rm R}} \succeq 0,
\end{gather}
and
\begin{gather}
\nonumber
\mu_2 :=
\max_{\rho_{\rm R}} \W\left[\tr_{{\rm R} \backslash \rmt}(\rho_{{\rm R}})\right], 
\text{ s.t. }  \tr_{{\rm R}\backslash \rms_{i}}(\rho_{{\rm R}}) = \sigma_{\rms_{i}} \forall\,\, \rms_{i}\in\calS, \\
\rho_{{\rm R}} \succeq 0.
\end{gather}
We have that $\mu_1 \le \mu_2$ since the former optimization has more constraints. However, note that $\rmt$, $\rms_i$ are subsystems in $\rms$, and
\begin{equation}
    \tr_{{\rm R} \backslash \rms_i} = \tr_{\rms \backslash \rms_i}\circ(\tr_{{\rm R} \backslash \rms}),\quad
    \tr_{{\rm R} \backslash \rmt} = \tr_{\rms \backslash \rmt}\circ(\tr_{{\rm R} \backslash \rms}).
\end{equation}
Hence, we can rewrite the latter problem as
\begin{gather}
\max_{\rho_{\rm R}} \W\left[\tr_{\rms \backslash \rmt}(\rhos)\right], 
\text{ s.t. }  \tr_{\rms \backslash \rms_{i}}(\rhos) = \sigma_{\rms_{i}} \forall\,\, \rms_{i}\in\calS, \,\,
\rho_{{\rm R}} \succeq 0.
\end{gather}
But now we see that the objective function and marginal constraints depend only on the subsystem $\rms$ of ${\rm R}$ and because partial trace is a positivity-preserving map, we can replace the last constraint with $\rhos \succeq 0$ and the optimization over $\rho_{{\rm R}}$ with the optimization over $\rhos$. Thus, we have that
\begin{equation}
    \mu_1 \le \mu_2 = \lambda < 0.
\end{equation}
This means we can extend any metatransitivity example to more parties as long as the additional constraints have a compatible global state.

\section{Local compatibility implies joint compatibility for classical-quantum marginals}\label{app.LocalCompJoint}

Here we show that in the tripartite case, for two classical-quantum states that overlap in a classical subsystem (i.e., its density matrix is diagonal in the computational basis), then compatibility in the overlapping subsystem leads to joint compatibility. We show this by constructing one of the possible global states.

Let $\{\ket{x}:x = 1,\ldots, d\}$ be an orthonormal basis for a $d$-dimensional Hilbert space. 
Consider the following bipartite states with local dimension $d$:
\begin{align}
   \sigma_{\rab} = \sum_{i=1}^{q} \sigma^{i}_{\ra} \otimes \sum_{x=1}^{d} \beta^{i}_{x} \proj{x}, \\
   \tau_{\rbc} = \sum_{j=1}^{r} \sum_{x=1}^{d} \widetilde{\beta}^{j}_{x} 
    \proj{x} \otimes \tau^{j}_{\rc}.
\end{align} 
for some $\sigma^{i}_{\ra}, \tau^{j}_{\rc} \geq 0$  and $\tr(\sigma^{i}_{\ra}) = \tr(\tau^{j}_{\rc})  =1$.
This requires $\beta^{i}_{x} \geq 0$ and $\widetilde{\beta}^{j}_{x} \geq 0$.
If $\sigma_{\rab}$ and $\tau_{\rbc}$ are compatible in ${\rm B}$ then we have that
\begin{equation}
\rho_{\rb} = \sum_{x=1}^{d} \rho_{\rb,x} \proj{x}, \quad
    \rho_{\rb,x} = \sum_{i=1}^{q} \beta^{i}_{x} = \sum_{j=1}^{r} \widetilde{\beta}^{j}_{x}, \forall x.
\end{equation}
Now we can introduce $\beta^{ij}_{x}$ such
that 
$\sum_{j} \beta^{ij}_{x} = \beta^{i}_{x}$ and $\sum_{i} \beta^{ij}_{x} = \widetilde{\beta}^{j}_{x}$. 
Then we can choose
\begin{equation}
    \beta^{ij}_{x} =
    \begin{cases}
        \dfrac{\beta^{i}_{x}\widetilde{\beta}^{j}_{x}}{\rho_{\rb,x}}, & \mbox{if } \rho_{\rb,x} \ne 0.\\
        0, & \mbox{otherwise.} 
    \end{cases}
\end{equation}

Then we can construct the tripartite state
\begin{align}
    \nonumber
    \rho_{\rabc} &= \sum_{ij} \sigma^{i}_{\ra} \otimes \sum_{x} \beta^{ij}_{x} \proj{x} \otimes \tau^{j}_{\rc}, \\
    &= \sum_{i,j,x} \beta^{ij}_{x} \left( \sigma^{i}_{\ra} \otimes \proj{x} \otimes \tau^{j}_{\rc}\right).
\end{align}
which is a valid density operator since this is a convex mixture of unit-trace, positive semidefinite operators.

The result can be easily extended to the multipartite case for marginal states that form a tree graph and where all overlapping subsystems are classical.

\section{Uniqueness of symmetric extensions for the Choi states of self-complementary qubit operations}
\label{app.uniqueSymExtSelfComp}

Here we will prove that the Choi state of a qubit self-complementary operation in \cref{eq.selfcompChoi} has a unique and pure symmetric extension. Our approach will be to consider an arbitrary
qubit tripartite state $\rho_{\rabc} = \sum_{i} c_{i} \proj{\Psi_i}$
and determine the form the eigenstates $\ket{\Psi_i}$ must take to produce the correct marginals in $\rab$ and $\rac$.

It is useful to observe that Eq.~(\ref{eq.selfcompChoi}) can be written as $\sigma_{\rab} = \proj{\varphi_0} + \proj{\varphi_1}$ where
\begin{align}
\ket{\varphi_0} &= a_0 \ket{00} + \tfrac{b}{\sqrt{2}}\ket{11}, &
\ket{\varphi_1} &= a_1 e^{it} \ket{01} + \tfrac{b}{\sqrt{2}}\ket{10},
\end{align}
are the unnormalized eigenvectors. This means we may examine separately the contributions to $\sigma_{\rab}$ in the orthogonal subspaces $\mathcal{V}_0 = \mathrm{span}\{\ket{00},\ket{11}\}$ 
and $\mathcal{V}_1 = \mathrm{span}\{\ket{01},\ket{10}\}$.

First, let us consider the contribution from $\mathcal{V}_0$ to the $\rac$ marginal state. It has the general form
\begin{equation}
    \alpha_{0} \ket{000} + \beta_{0} \ket{101} + \alpha_{1}\ket{010} + \beta_{1}\ket{111},
\end{equation}
where to match $\ket{\varphi_0}$ after we trace out system $\rb$, we require
\begin{align}
    a_0^2 &= \alpha_0^2+ \alpha_1^2 = \lVert \vec{\alpha} \rVert^2, &
    \tfrac{b^2}{2} &= \beta_0^2+ \beta_1^2 = \lVert \vec{\beta} \rVert^2,
\end{align}
where $\vec{\alpha} = (\alpha_0,\alpha_1)$ and $\vec{\beta} = (\beta_0, \beta_1)$. We see that
\begin{equation}
    \tfrac{a_{0}b}{\sqrt{2}} = \alpha_0^*\beta_0 + \alpha_1^*\beta_1 = \innerprod{\vec{\alpha}}{\vec{\beta}}.
\end{equation}
This means that $\vec{\alpha}$ and $\vec{\beta}$ saturate the Cauchy-Schwarz inequality
\begin{equation}
   \left\lvert \innerprod{\vec{\alpha}}{\vec{\beta}}\right\rvert^2 =  \lVert \vec{\alpha}\rVert^2 \lVert \vec{\beta}\rVert^2, 
\end{equation}
which implies that $\vec{\alpha}$ and $\vec{\beta}$ must be linearly dependent, i.e.,
$\gamma = \tfrac{\alpha_1}{\alpha_0} = \tfrac{\beta_1}{\beta_0}.$
This suggests that the contribution should have the form
$\alpha_{0}(\ket{000} + \gamma\ket{010}) + \beta_{0} (\ket{101} + \gamma\ket{111}).$
However, since the $\rab$ and $\rac$ must be the same state, we need to add terms to make it symmetric with respect to $\rb$ and $\rc$:
$\alpha_{0}(\ket{000} + \gamma\ket{010} + \underline{\gamma \ket{001}} ) + \beta_{0} (\ket{101} + \underline{\ket{110}} + \gamma\ket{111}).$
Finally, we notice this is a superposition of terms with and without the factor $\gamma$ that can be independent, so the contributions from $\mathcal{V}_0$ have the form
\begin{align}
\label{eq.V0contrib}
    \ket{\Psi_0} \in \{& \alpha_0\ket{000} + \beta_0(\ket{101} + \ket{110}), \nonumber \\
    &\, \alpha_0(\ket{010} + \ket{001}) + \beta_0 \ket{111} \}.
\end{align}

Next we consider the contribution from $\mathcal{V}_1$. But we observe that $\ket{\varphi_1}$ essentially has the same form as $\ket{\varphi_0}$ so we can make the same argument just by substituting  
\begin{equation}
    (\ket{00},\ket{11}, a_{0}, \alpha_0,\beta_0, \gamma) \mapsto (\ket{01}, \ket{10}, a_{1}e^{it},\alpha_1,\beta_1, \delta).
\end{equation}
Thus, the contribution from $\mathcal{V}_1$ can be immediately written as
\begin{align}
\label{eq.V1contrib}
    \ket{\Psi_1} \in \{& \alpha_1\ket{011} + \beta_1(\ket{101} + \ket{110}), \nonumber \\
    &\, \alpha_1(\ket{010} + \ket{001}) + \beta_0 \ket{100} \}.
\end{align}
Now we will combine the contributions from $\mathcal{V}_{0}$ and $\mathcal{V}_{1}$. Observe that the respective first states in \cref{eq.V0contrib} and \cref{eq.V1contrib} have a common term $(\ket{110} + \ket{101})$. This term needs $\ket{000}$  and $\ket{011}$ to produce the correct marginals in $\mathcal{V}_0$ and
$\mathcal{V}_1$, respectively. This suggests that they should appear together and with $\beta_0 = \beta_1$ we have the candidate eigenstate
\begin{equation}
    \ket{\Psi} = \alpha_0 \ket{000} + \alpha_1\ket{011} + \beta_0(\ket{110} + \ket{101})
\end{equation}
similarly, by taking the respective second states in \cref{eq.V0contrib} and \cref{eq.V1contrib}, we see that they share the term $(\ket{010} + \ket{001})$.  This term needs $\ket{111}$ and $\ket{100}$ to produce the correct marginals in $\mathcal{V}_0$ and
$\mathcal{V}_1$, respectively. This gives the other candidate eigenstate 
\begin{equation}
    \ket{\overline{\Psi}} = \alpha_0' (\ket{010} + \ket{001}) + \beta_0' \ket{111} + \beta_1' \ket{100}).
\end{equation}
At this point, the global state is in $\text{span}\{\ket{\Psi},\ket{\overline{\Psi}}\}$, so the only two eigenstates can be written as  $\ket{\Psi_1} =\mu \ket{\Psi}+\nu \ket{\overline{\Psi}}$ and $\ket{\Psi_2} =\nu \ket{\Psi}-\mu \ket{\overline{\Psi}}$, where $|\mu|^2+|\nu|^2=1$. We first consider when the global state is rank-2: $\rho_{\rabc}=c\ketbra{\Psi_1}{\Psi_1}+(1-c)\ketbra{\Psi_2}{\Psi_2},c\in(0,1)$. To satisfy $\tr (\sigma_{\rab}\ketbra{10}{10})=\tr(\sigma_{\rab}\ketbra{11}{11})$, we have $\beta_0'=\beta_1'$. Here we define
\begin{equation}
    \Tilde{\rho}_{\rab} =\tr_\rc(\rho_{\rabc})     =\begin{pmatrix}
\Tilde{\rho}_{11} & \Tilde{\rho}_{12}  &  \Tilde{\rho}_{13}  & \Tilde{\rho}_{14}  \\
\Tilde{\rho}_{21} & \Tilde{\rho}_{22}  &  \Tilde{\rho}_{23}  & \Tilde{\rho}_{24} \\
\Tilde{\rho}_{31} & \Tilde{\rho}_{32}  &  \Tilde{\rho}_{33}  & \Tilde{\rho}_{34} \\
\Tilde{\rho}_{41} & \Tilde{\rho}_{42}  &  \Tilde{\rho}_{43}  & \Tilde{\rho}_{44} 
\end{pmatrix},
\end{equation}
Looking at the subspace spanned by $\ket{01}$ and $\ket{10}$, we have
\begin{align}
\label{Eq:01_10_subspace}
        \Tilde{\rho}_{22} &= |\vec{v_1}|^2, &
        \Tilde{\rho}_{23} &= \vec{v_1}\cdot \vec{v_2}^*, \nonumber \\
        \Tilde{\rho}_{32} &= \vec{v_1}^*\cdot \vec{v_2},&
        \Tilde{\rho}_{33} &= |\vec{v_2}|^2,
\end{align}
where 
\begin{align}
    v_1 &=\left(\alpha_0'\sqrt{(1-c)|\mu|^2+c|\nu|^2},\alpha_1\sqrt{c|\mu|^2+(1-c)|\nu|^2} \right)\Tp, \nonumber \\
    v_2 &=\left(\beta_0'\sqrt{(1-c)|\mu|^2+c|\nu|^2},\beta_0\sqrt{c|\mu|^2+(1-c)|\nu|^2}\right)\Tp.
\end{align}
However, comparing this with the corresponding sub-matrix in $\sigma_{\rab}$, the two vectors should saturate the Cauchy–Schwarz inequality, which implies $\vec{v_1}=\eta\vec{v_2}$ for some constant $\eta$. Thus, we have $\frac{\alpha_0'}{\beta_0'}=\frac{\alpha_1}{\beta_0}$. Similarly for the subspace spanned by $\ket{00}$ and $\ket{11}$, we have
\begin{align}
\label{Eq:00_11_subspace}
        \Tilde{\rho}_{11} &= |\vec{v_3}|^2,&
        \Tilde{\rho}_{14} &= \vec{v_3}\cdot \vec{v_4}^*, \nonumber \\
        \Tilde{\rho}_{41} &= \vec{v_3}^*\cdot \vec{v_4}, &
        \Tilde{\rho}_{44} &= |\vec{v_4}|^2,
\end{align}
where 
\begin{align}
  v_3 &= \left(\alpha_0'\sqrt{(1-c)|\mu|^2+c|\nu|^2},\alpha_0\sqrt{c|\mu|^2+(1-c)|\nu|^2}\right)\Tp, \nonumber \\ 
  v_4 &= \left(\beta_0'\sqrt{(1-c)|\mu|^2+c|\nu|^2},\beta_0\sqrt{c|\mu|^2+(1-c)|\nu|^2} \right)\Tp,
\end{align}
so $\frac{\alpha_0'}{\beta_0'}=\frac{\alpha_0}{\beta_0}$. However, this means $\alpha_0=\alpha_1$, which will imply that $\Tilde{\rho}_{11}=\Tilde{\rho}_{22}$. However, this means that $|a_0| = |a_1|$  and this leads to a separable $\sigma_{\rab}$. Therefore, for entangled $\sigma_{\rab}$, the global state cannot be rank-2. For all the possible rank-1 global states $\rho_{\rabc}=\ketbra{\Psi_1}{\Psi_1}$, we can use the same argument above (setting $c = 1$) to exclude the situation when $\nu \neq 0$. As a result, $\rho_{\rabc}=\ketbra{\Psi}{\Psi}$ is the unique global state.
